\documentclass[journal,twocolumn]{IEEEtran}

\usepackage[utf8]{inputenc} 
\usepackage[T1]{fontenc}
\usepackage{url, ifthen, cite}
\usepackage[cmex10]{amsmath}
\usepackage{graphicx, graphics, balance, dsfont}
\usepackage{amssymb, amsfonts, amsthm, bm}
\usepackage{epsfig, tikz, standalone}

\newcommand{\PP}{\mathbb{P}}
\newcommand{\E}{\mathbb{E}}
\newcommand{\al}{\alpha}
\newcommand{\R}{\mathbb{R}}
\newcommand{\la}{\lambda}

\newcommand{\eps}{\epsilon}

\newcommand{\sinr}{{\sf SINR}}
\newcommand{\e}{{\sf E}}
\newcommand{\sar}{{\sf SAR}}
\newcommand{\mpe}{{\sf MPE}}

\newcommand{\qh}{{\bf h}}

\newcommand{\qw}{{\bf w}}
\newcommand{\qx}{{\bf x}}

\newcommand{\qz}{{\bf z}}

\newcommand{\qA}{{\bf A}}
\newcommand{\qB}{{\bf B}}

\newcommand{\qI}{{\bf I}}

\newcommand{\qQ}{{\bf Q}}

\newcommand{\qW}{{\bf W}}
\newcommand{\qX}{{\bf X}}

\newcommand{\qzero}{{\bf 0}}

\newcommand{\qtheta}{{\boldsymbol \theta}}
\newcommand{\tr}{\mbox{tr}}

\newtheorem{theorem}{Theorem}
\newtheorem{corollary}{Corollary}
\newtheorem{lemma}{Lemma}

\newcommand{\be}{\begin{equation}} \newcommand{\ee}{\end{equation}}
\newcommand{\bea}{\begin{eqnarray}} \newcommand{\eea}{\end{eqnarray}}

\allowdisplaybreaks

\begin{document}
\title{Design and Analysis of SWIPT\\ with Safety Constraints}

\author{Constantinos Psomas, \IEEEmembership{Senior Member, IEEE}, Minglei You, Kai Liang,\\ Gan Zheng, \IEEEmembership{Fellow, IEEE}, and Ioannis Krikidis, \IEEEmembership{Fellow, IEEE}

\thanks{C. Psomas and I. Krikidis are with the Department of Electrical and Computer Engineering, University of Cyprus, Nicosia, Cyprus (email: \{psomas, krikidis\}@ucy.ac.cy).}
\thanks{M. You is with the Department of Electrical and Electronic Engineering, University of Nottingham, Nottingham, NG7 2RD, UK (email: minglei.you@nottingham.ac.uk)}
\thanks{K. Liang is with the State Key Laboratory of Integrated Service Networks, Xidian University, Xi’an 710071, China (email: kliang@xidian.edu.cn).}
\thanks{G. Zheng is with the Wolfson School of Mechanical, Electrical and Manufacturing Engineering, Loughborough University, Leicestershire, LE11 3TU, UK (email: g.zheng@lboro.ac.uk).}
\thanks{The work of C. Psomas and I. Krikidis was supported by the
European Research Council (ERC) through the European Union's Horizon
2020 Research and Innovation Programme under grant agreement 819819 (APOLLO). The work of K. Liang was supported by the National Natural Science Foundation of China (61901317), the Fundamental Research Funds for the Central Universities (JB190104) and the Joint Education Project between China and Central-Eastern European Countries (202005). The work of G. Zheng was supported in part by the UK Engineering and Physical Sciences Research Council (EPSRC) under grants EP/N007840/1.}}

\maketitle

\begin{abstract}
Simultaneous wireless information and power transfer (SWIPT) has long been proposed as a key solution for charging and communicating with low-cost and low-power devices. However, the employment of radio frequency (RF) signals for information/power transfer needs to comply with international health and safety regulations. In this paper, we provide a complete framework for the design and analysis of far-field SWIPT under safety constraints. In particular, we deal with two RF exposure regulations, namely, the specific absorption rate (SAR) and the maximum permissible exposure (MPE). The state-of-the-art regarding SAR and MPE is outlined together with a description as to how these can be modeled in the context of communication networks. We propose a deep learning approach for the design of robust beamforming subject to specific information, energy harvesting and SAR constraints. Furthermore, we present a thorough analytical study for the performance of large-scale SWIPT systems, in terms of information and energy coverage under MPE constraints. This work provides insights with regards to the optimal SWIPT design as well as the potentials from the proper development of SWIPT systems under health and safety restrictions.
\end{abstract}


\section{Introduction}
Wireless technologies are an important part of modern society, with Cisco expecting that global mobile subscribers will reach 5.7 billion by 2023, which will correspond 71 percent of the global population \cite{CISCOAIRWP2020}. Traditionally, the focus of wireless communications is mainly on how to improve the efficiency of information transfer. Nonetheless, with the development of wireless systems employing a massive number of devices, such as sensors and actuators, recently this focus has also been shifted towards energy sustainability. In particular, the fact that radio frequency (RF) signals can also convey energy apart from information, the concept of wireless power transfer (WPT) and, in particular, of simultaneous wireless information and power transfer (SWIPT) is considered as a very promising and enabling technology for the realization of such future systems \cite{Rui-TCOM-17}. The key idea of SWIPT is to exploit the received RF signals in order to extract not only information but energy as well. The extraction can be done by separating the energy harvesting (EH) and information decoding operations either in time, in power or in space \cite{KRI}. In contrast to conventional EH techniques (e.g. from renewable sources), EH with SWIPT can be a dedicated, continuous, controllable and on-demand process. The EH in this case, is achieved through the employment of a rectifying antenna (rectenna) that converts the received RF signal to direct current (DC) \cite{Rui-TCOM-17,KRI}.

In light of the wide and increasing usage of wireless devices, there is a growing concern with regards to the RF radiation brought about by multiple and concurrent wireless transmissions. Some studies have also revealed the potential biological hazard in relation to the RF radiation, including metabolic changes in brain and carcinogenic effects \cite{volkow2011effects, international2011iarc}. As such, international health and safety regulations have been put in place in order to regulate and limit the level of RF exposure to humans \cite{ALOUINI}. Two widely adopted regulations/measures on RF exposure are the so-called maximum permissible exposure (MPE) and specific absorption rate (SAR). The MPE or power density (measured in W/m$^2$), defines the highest level of electromagnetic radiation (EMR) in a specific area that will not incur any health/biological effect. On the other hand, SAR (measured in W/kg) is a localized metric and defines the maximum level of absorbed power in a unit mass of human tissue.

Due to the vital impact on applications with WPT, MPE is concerned by some relevant studies \cite{safe_charging,safe_distributed,scheduling_charging,Dai-InfoCom-18}. The problem of scheduling the power chargers is investigated in \cite{safe_charging}, where the charging utility for all rechargeable devices is maximized with a constraint on EMR. The works in \cite{safe_distributed} and \cite{scheduling_charging} deal with the problem of maximizing the harvested energy and wireless charging tasks scheduling, respectively, when transmitted signals guarantee a well-defined EMR constraint. The work in \cite{Dai-InfoCom-18} builds an empirical probabilistic wireless charging and EMR model and then formulates an optimization problem to maximize the charging utility of all EH devices with consideration on the EMR constraint. For short distances (i.e., distances less than $20$ cm), the SAR measure dominates the RF exposure. Therefore, in this case, SAR becomes a more critical factor than the MPE for the design of efficient communications. Some regulatory agencies have established limitations on the body SAR exposure. For instance, the Federal Communications Commission (FCC) enforces a SAR limitation of $1.6$ W/kg averaged over one gram of tissue on the partial body exposure \cite{FCC-01}, and the Comit\'e Europ\'een de Normalisation \'Electrotechnique adopts a similar limitation of $2$ W/kg averaged over $10$ grams of tissue on the SAR measurements \cite{ULC}. In addition, SAR needs to account for various exposure constraints on the whole body, partial body, hands, wrists, feet, ankles, etc., with various measurement limitations according to FCC regulations \cite{FCC-01}. Thus, multiple SAR constraints are needed even for a single transmit device. For instance, the iPhone 12 Pro Model A2407 has a whole body SAR of $1.18$ W/kg and a head SAR of $1.14$ W/kg \cite{iphone}.

Although the subject of an RF exposure constraint (i.e., the SAR constraint) has an important impact on the design of wireless communication systems, few existing works in the literature have considered SAR regulations. The SAR exposure limitation can be easily guaranteed in single-antenna systems by introducing an additional transmit power constraint (i.e., setting the transmit power below a required threshold). In the meantime, the exploitation of multiple antennas provides significant benefits in improving the throughput of a wireless communication system, but it also poses potential challenges during the design due to RF radiation restrictions. For example, by exploiting advanced signal processing techniques such as beamforming in multi-antenna systems, the pattern of the RF signals is manipulated in such a way as to increase the performance, while it also makes the associated SAR analysis more complicated. Measurements and simulations are carried out in \cite{Murch-04} and demonstrate that SAR is a function of the phase difference between two transmit antennas. The SAR reduction and modeling in multi-antenna systems is studied in \cite{Qiang-07,Mahmoud-08}. The SAR constraints are integrated into the transmit signal design for a multiple-input multiple-output (MIMO) uplink channel in \cite{Hochwald-12}, where the quadratic model for the SAR measurements is first proposed. In \cite{Hochwald-14}, a SAR code is proposed to improve the conventional Alamouti space-time code under SAR constraints. It is also revealed that the SAR measurement is a function of the quadratic form of the transmitted signal with the SAR matrix. The SAR-aware beamforming and transmit signal covariance optimization methods are presented in \cite{Ying-CISS-13} and \cite{Ying-Globecom-13}. Capacity analysis with multiple SAR constraints on single-user MIMO systems is intensively examined in \cite{Ying-TWC-13}. Sum-rate analysis for a multi-user MIMO system with SAR constraints is performed in \cite{Ying-TWC-17}, with both perfect and statistical channel state information (CSI). From an information theory perspective, SAR-constrained multi-antenna transmit covariance optimization can be seen as the classical MIMO channel capacity optimization problem subject to generalized linear transmit covariance constraints \cite{Rui-TIT-12}.

The consideration of SAR and MPE constraints for the design of WPT or SWIPT systems is an under-explored research area \cite{Rui-TCOM-17}. Even though SWIPT corresponds to a controlled transmission of RF radiation to communicate as well as energize, it may significantly contribute to the electromagnetic pollution (electrosmog). However, few works in the literature discuss the integration of SAR and MPE with WPT \cite{safe_charging, safe_distributed, scheduling_charging, Dai-InfoCom-18} or SWIPT \cite{Zhang-20}. As such, this paper provides a complete framework for the study of SWIPT, under both SAR and MPE constraints. We first present the mathematical modeling of the SWIPT technology as well as the modeling aspects of the two safety metrics. Then, we describe methodologies for the design of simple but also complex large-scale SWIPT systems under safety constraints; the methodologies are general and can be adapted to any communication scenario. Specifically, the contributions are as follows:
\begin{itemize}
  \item We first introduce and formulate the beamforming optimization problem to maximize the harvested power in a multiple-input single-output (MISO) downlink system, subject to SAR constraints and quality-of-service requirements. We derive the  beamforming solution  by leveraging semidefinite programming and rank relaxation, and prove that this method always achieves the optimal solution.
  \item Next, a robust beamforming design is proposed  subject to SAR constraints at the receivers. In practical systems, the CSI is usually measured or estimated, while there are many factors contributing to errors, e.g. quantization errors \cite{you2021data}. In such cases, the constraints such as the user end performance, measured by the signal-to-interference-plus-noise ratio (SINR) are characterized in a statistical instead of a deterministic manner and may be violated. This challenge is traditionally addressed via robust beamforming. The major solutions to robust beamforming is to provide the worst-case guarantees or probabilistic performance guarantees, such as the semidefinite relaxation (SDR) \cite{zheng2008robust}, Bernstein-Type Inequality (BTI) method \cite{wang2014outage}, and Large Deviation Inequality (LDI) method \cite{yuan2019joint}. However, these solutions require high computational complexity, which incur large latency and they are over-preservative to account for the worst cases. The proposed design is based on a low-complexity unsupervised deep learning approach with the data augmentation technique. Our results show that it achieves significant improvement and outperforms the BTI method. 
  \item The MPE constraint is considered in SWIPT networks from a macroscopic point-of-view through the employment of stochastic geometry. Closed-form analytical expressions are derived for the probability of satisfying the MPE constraint, the information coverage probability, the energy coverage probability as well as their joint probability. The performance with and without the constraints is considered and it is shown how each metric is affected by this restriction. Moreover, we study the system's performance for different frequency bands and our results show that higher frequencies decrease the levels of RF exposure in the network.
\end{itemize}
The rest of this paper is organized as follows: Section \ref{modeling} describes the modeling of SWIPT and of the considered safety regulations. In Section \ref{swipt_sar}, the proposed design of robust beamforming for SWIPT under SAR constraints is provided together with simulation results. Section \ref{swipt_mpe} presents the performance analysis of a large-scale SWIPT network under MPE constraints. The paper concludes with Section \ref{conclusion}.

{\it Notation}: Lower and upper case boldface letters denote vectors and matrices, respectively; $[\cdot]^\dag$ is the Hermitian transpose operator; $\PP\{X\}$ and $\E\{X\}$ represent the probability and expectation of $X$, respectively; $\R^n$ denotes the $n$-dimensional Euclidean space; $\Gamma(\cdot)$ and $\Gamma(\cdot,\cdot)$ denote the complete and upper incomplete gamma function, respectively \cite{GRAD}; $\mathrm{B}(\cdot,\cdot)$ denotes the beta function \cite{GRAD}; $\jmath = \sqrt{-1}$ is the imaginary unit; $\tr(\qA)$ gives the trace of the square matrix $\qA$; $\qA \succeq \qzero$ means that the matrix $\qA$ is positive semidefinite; $\Im\{x\}$ and $\Re\{x\}$ return the imaginary and real part of $x$, respectively; $U(a,b)$ denotes the uniform distribution in the interval $[a,b]$.

\section{Modeling SWIPT and Safety Constraints}\label{modeling}

\subsection{Information and Power Transfer}
Throughout this paper, we consider multi-antenna transmitters and single-antenna receivers. Also, each receiver has SWIPT capabilities, i.e. it can decode the information but also harvest energy from the received signal simultaneously. The SWIPT technique is employed with the power splitting (PS) method such that the received signal is split into two parts: one is converted to a baseband signal for information decoding and the other is directed to the rectenna for EH and storage \cite{ZHA}. This is a mature SWIPT technique that does not require strict time synchronization between information and power transfer. Let $\rho \in (0,1)$ denote the PS parameter at a receiver. Then, $100\rho\%$ of the received power is used for decoding, while the remaining power is directed to the EH circuit. During the baseband conversion phase, additional circuit noise is present due to the phase-offsets and the circuit's non-linearities, which is modeled as an additive white Gaussian noise (AWGN) with zero mean and variance $N_C$.

Therefore, based on the PS technique considered, the SINR at a receiver can be written as
\begin{align}
  \sinr = \frac{\rho S}{\rho(N_0 + I)+N_C},
\end{align}
where $S$ is the received power of the signal of interest, $I$ is the power of the interference and $N_0$ is the variance of the AWGN component of the received signal. On the other hand, since $100(1 - \rho)\%$ of the received energy is used for rectification, the instantaneous energy harvested at a receiver is modeled by the following non-linear function, which refers to a specific excitation signal\footnote{The provided mathematical framework is not limited to this model and can be easily adapted to consider other non-linear functions, e.g. the sigmoid model \cite{NG}.} \cite{Yunfei-17}
\begin{align}\label{eqn:nonlinear}
  \e = \frac{\bar a (1-\rho) P_r + \bar b}{(1-\rho) P_r + \bar c}-\frac{\bar b}{\bar c},
\end{align}
where $P_r$ is the aggregate received signal power at the receiver and $\bar a$, $\bar b$, $\bar c$ are parameters determined by the rectification circuit through curve fitting. These parameters fully characterize the non-linear behaviour of the rectifying circuit including the maximum harvesting value, the sensitivity and slope of the output power \cite{Yunfei-17}. In other words, depending on how a rectifier is designed will correspond to a different set of parameters; in this work, we will consider $\bar a = 2.463, \bar b = 1.635$, and $\bar c = 0.826$ \cite{Yunfei-17}. Note that, in general, $\e$ (measured in Watts) should be a function of the received signal rather than just its power. In this paper, we adopt a simplified model to highlight the dependency on the power of the energy signal only, which will be discussed in the numerical results.

\subsection{Safety Constraints in Wireless Networks}
The SAR metric quantifies the deposited microwave energy at a specific point on the human body. It corresponds to the rate of energy absorption per unit mass at a specific location in the tissue \cite{Hochwald-14}. As such, SAR is a function of the induced electric field $E$ (measured in V/m), the electrical conductivity of the tissue $\sigma$ at the specific point (measured in S/m) as well as the tissue's density $\eta$ (measured in kg/m$^3$). This relation can be expressed as
\begin{align}
\sar = \frac{\sigma |E|^2}{2\eta}.
\end{align}
In order to integrate the SAR limitations into the design of SWIPT, we use a quadratic form of the transmitted signal to model the pointwise SAR value with multiple transmit antennas \cite{Hochwald-14}. This model is based on experimental studies, which showed that SAR depends significantly on the phase difference between the antennas \cite{CHIM} and thus can be characterized by a sinusoidal function of the phase difference \cite{Hochwald-12}. Let $\qx$ denote the transmitted signal and $\qQ = \E\{\qx\qx^\dag\}$ its covariance. Then, SAR can be modeled as a quantity averaged over the transmit signals with a time-averaged quadratic constraint given by
\begin{align}\label{key}
\sar &= \mathbb{E}\{\tr(\qx^\dag \qA \qx)\} = \tr(\qA \qQ) \le P,
\end{align}
where $\qA$ is the SAR matrix and $P$ is the SAR limit. The dependence of the SAR measurements on the transmitted signals can be fully described by the SAR matrix, where the entries of this matrix have units of kg$^{-1}$. Note that the SAR matrices are positive-definite conjugate-symmetric matrices, since the SAR measurements are always real positive numbers. The SAR matrix is obtained offline during SAR testing and it highly depends on the device's type of antennas, operating frequency, industrial design, etc. \cite{Hochwald-14}. Moreover, during testing, several measurements are taken for different operations and locations of the device and the one that provides the ``worst case'' is used for comparison with the SAR constraint.

Different to the SAR metric, which is a point quantity, the MPE (power density) quantifies the RF exposure over a specific area. The MPE corresponds to either the absorbed power density or the incident power density \cite{ALOUINI}. The latter is easier to be measured and thus international regulations are given in terms of the maximum incident power density values. The MPE is given by
\begin{align}
\mpe = \frac{|E|^2}{Z},
\end{align}
where $E$ [V/m] is the induced electric field and $Z = 377$ $\Omega$ is the impedance of free space \cite{ALOUINI}. In the context of wireless networks, the MPE can be evaluated by 
\begin{align}\label{mpe}
\mpe = \frac{P_t G}{4\pi d^2},
\end{align}
where $P_t$ is the transmitted power, $G$ is the antenna gain and $d$ is the distance of the measuring point from the center of the antenna \cite{FCC}.

The SAR and MPE metrics are equally important for the realization of safe (in terms of RF radiation) wireless networks, where each metric concerns a different aspect of the network. The SAR metric is considered for devices that are meant to be carried close to the body (e.g. mobile phones) \cite{Hochwald-14, CHIM} or to be ``worn'' on the body (e.g. implantable sensors) \cite{AK,SP}. Essentially, these devices will be operating at a distance of $20$ cm or less from a human body \cite{Hochwald-12}. Moreover, the devices need to abide by the FCC's SAR regulations \cite{FCC-01} and thus are tested before being commercially available. In the case of MPE, the devices under consideration are expected to be operating at a distance greater than $20$ cm from a human body (e.g. base stations) \cite{Hochwald-12}. Therefore, the MPE is evaluated based on a network of transmitting devices and all locations in the network area are required to satisfy the MPE constraint \cite{FCC}. It is clear that SAR focuses more on the device itself, whereas MPE is a network-wide safety constraint. This motivates the approach taken in the following two sections. In Section \ref{swipt_sar}, we consider a SAR-aware beamforming optimization for a simple point-to-multipoint SWIPT system and, in Section \ref{swipt_mpe}, we study the performance of a SWIPT system from a macroscopic point-of-view with MPE constraints.

\section{SWIPT with SAR Constraints:\\ Optimization of Transmit Beamforming}\label{swipt_sar}
In this section, we study SAR-aware transmit beamforming optimization with both perfect and statistical CSI.

\subsection{System Model and Problem Formulation}
We consider a MISO downlink system consisting of an $N_t$-antenna transmitter with total transmit power $P_t$ and $K$ single-antenna receivers that employ single-user detection, as shown in Fig. \ref{fig:sys}. The transmitted data symbol $s_k$ to receiver $k$ follows the Gaussian distribution with zero mean and unit variance (i.e., $\mathbb{E}\{\|s_k\|^2\}=1$), which is mapped onto the antenna array elements by the beamforming vector $\mathbf{w}_k \in \mathbb{C}^{N_t\times 1}$.	The signal design in terms of modulation, waveform and input distribution will also affect the efficiency of the RF-DC conversion \cite{CLE3,learning_signal}, but for simplicity, we do not consider these in the optimization.
	
We assume frequency non-selective block fading channels (i.e., the channel coefficients remain constant in each slot) with AWGN. Denote by $\mathbf{h}_k \in \mathbb{C}^{N_t\times 1}$ the fading coefficients between the transmitter and receiver $k$, which also captures the large-scale degradation effects such as path-loss and shadowing. The received baseband signal at the receiver $k$ can be expressed as
\begin{align}\label{sys1}
y_k = \underbrace{\mathbf{h}_k^\dag \mathbf{w}_k s_k}_{\textrm{Information signal}} + \underbrace{\sum_{j\neq k}\mathbf{h}_k^\dag \mathbf{w}_js_j}_{\textrm{Interference}} + n_k,
\end{align}
where $n_k$ denotes the AWGN component with zero mean and variance $N_0$.
	
The receivers harvest energy from the received RF signal based on the PS technique and so the SINR used for the data detection process at the $k$-th receiver is given by
\begin{equation}\label{eq gamma}
\Gamma_k = \frac{\rho_k |\mathbf{h}_k^\dag \mathbf{w}_k|^2}{\rho_k \left(N_0+\sum_{j\neq k} |\mathbf{h}_k^\dag\mathbf{w}_j|^2\right) + N_C},
\end{equation}
while the input to the RF-DC circuitry is
\begin{equation}\label{eq lambda}
\Lambda_k = (1-\rho_k) P_{r,k},
\end{equation}
where $P_{r,k}$ is the received power at receiver $k$, given by
\begin{equation}\label{eq:TotalPower}
P_{r,k} = \sum_{j=1}^K |\mathbf{h}_k^\dag \mathbf{w}_j|^2 + N_0.
\end{equation}
Finally, the $l$-th SAR constraint with a time-averaged quadratic constraint is
\begin{align}
\sar_l &= \E_{s_k}\left\{\tr\left(\sum_{k=1}^K s_k^\dag \qw_k^\dag \qA_l \mathbf{w}_k s_k\right)\right\}\nonumber\\
&= \sum_{k=1}^K \qw_k^\dag \qA_l \qw_k\le P_l,
\end{align}
where $\qA_l\succeq \qzero$ is the $l$-th SAR matrix and $P_l$ is the $l$-th SAR limit.

\begin{figure}\centering
	\includegraphics[width=0.65\linewidth]{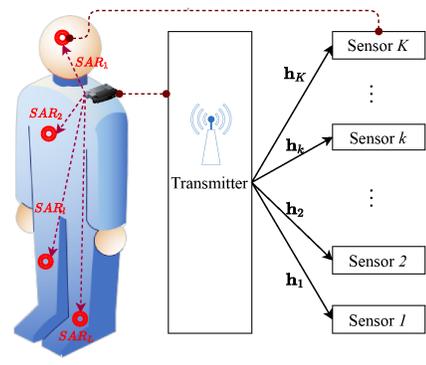}
	\caption{System model of SWIPT with SAR constraints.}\label{fig:sys}
\end{figure}

\subsubsection{Problem Formulation}
We formulate the problem as maximizing the harvested power at each user while satisfying the requirements of both SINR and total transmit power under the SAR constraints. To make the problem more tractable, we introduce a slack variable $\lambda$, and then the problem can be formulated as follows
\bea\label{eqn:prob:0}
\textbf{P1:}&& \max_{\{\qw_k, \rho_k\}} \lambda\\
\mbox{s.t.}&&
\frac{|\qh_{k}^\dag\qw_k|^2}{\sum\limits_{j=1,j \ne k}^K |\qh_k^\dag\qw_j|^2 + N_0 + \frac{N_C}{\rho_k}} \ge \gamma_k,\notag\\
&& (1-\rho_k) \left(\sum_{j=1}^K |\qh_j^\dag\qw_k|^2 + N_0 \right) \ge \lambda,\notag\\
&& 0\le \rho_k\le 1, \forall k, \notag \\
&& \sum_{k=1}^K \|\qw_k\|^2 \le P_t, \notag\\
&& \sum_{k=1}^K \qw_k^\dag \qA_l \qw_k \leq P_l, \forall l=1,\dots, L, \notag
\eea
where $\gamma_k$ is the SINR threshold at the $k$-th receiver. Clearly, \textbf{P1} is a non-convex problem because of both the SINR and EH constraints and thus is difficult to solve. In the next subsection, we develop an efficient convex optimization based algorithm that jointly optimizes the beamforming vectors and the PS parameters.
	
\subsubsection{The Optimal Solution using SDP}
We adopt the semidefinite programming (SDP) approach with rank relaxation to solve \textbf{P1}. We first define a matrix variable $\qW_k = \qw_k\qw_k^\dag$, and introduce $\tilde\lambda = \sqrt{\lambda}$ to recast the problem \textbf{P1} as follows
\bea\label{eqn:prob:1}
\textbf{P2:} && \max_{\{\qW_k, \rho_k, \tilde\lambda\}} \tilde\lambda\\
\mbox{s.t.} &&
\frac{\tr(\qh_k\qh_k^\dag\qW_k)}{\sum\limits_{j=1}^K \tr(\qh_k\qh_k^\dag\qW_j) + N_0 + \frac{N_C}{\rho_k}} \ge \frac{\gamma_k}{1+\gamma_k},\notag\\
&& \sum_{j=1}^K \tr(\qh_k\qh_k^\dag\qW_j) + N_0 \ge \frac{\tilde\lambda^2}{1-\rho_k},\notag\\
&& 0\le \rho_k\le 1, \qW_k\succeq \qzero, \forall k,\notag\\
&& \sum_{k=1}^K \tr(\qW_k) \le P_t,\notag\\
&& \sum_{k=1}^K \tr(\qA_l\qW_k) \leq P_l, \forall l=1,\dots,L.\notag
\eea
Note that the original objective value should be $\tilde\lambda^2$. The advantage of the problem \textbf{P2} is that it is convex, because it is linear in all $\{\qW_k\}$ and both terms $\frac{1}{\rho_k}$ and $\frac{1}{1-\rho_k}$ are convex in $\rho_k>0$. It can be efficiently solved using numerical software packages such as CVX \cite{cvx}. Once \textbf{P2} is optimally solved, if the resulting solutions $\{\qW_k\}$ are all rank-$1$, they are the exact optimal solutions; otherwise, the solutions only provide a lower bound for the minimum required transmit power. However, whether the SDP with rank relaxation can generate the optimal solution highly depends on the problem structure. With the additional SAR constraints, it is unknown whether this property remains true for the problem \textbf{P2}. In the following theorem, we show that this is indeed the case.

\begin{theorem}\label{theo2}
The optimal solution to \textbf{P2} satisfies $\mbox{rank}(\qW_k)=1, \forall k$, i.e., the SDP relaxation is tight, and the optimal solution to the problem \textbf{P1} can be recovered from $\{\qW_k\}$ via the eigenvalue decomposition.
\end{theorem}

\begin{proof}
The proof is given in Appendix \ref{prf_thm1}.
\end{proof}
	
\subsection{Robust Beamforming Solution using Deep Learning}\label{section robust design}
It is noticed that the formulation of problem \textbf{P1} and \textbf{P2} are based on the assumption of perfect CSI. However, in practical systems, it is hard for the transmitter to obtain perfect CSI, which could be subject to CSI estimation errors. Without loss of generality, the general relation between the estimated CSI $\hat \qh_k$ and the actual CSI $\qh_k$ is described as follows
\be
\qh_k = \hat \qh_k + \Delta \qh_k, \Delta \qh_k\in \mathcal{H}_k,
\ee
where $\Delta \qh_k$ denotes the channel estimation error and $\mathcal{H}_k$ denotes the set of all potential channel estimation errors. In practice, it is not feasible to know the exact channel estimation error $\Delta \qh_k$ in prior, since $\Delta \qh_k$ is usually not a deterministic value but a random variable. The statistical characterization regarding $\Delta \qh_k$ can be available, for example via measurements and calibrations. Here, we model $\Delta \qh_k$ by using the complex Gaussian distribution with zero mean and variance $\sigma_h^2$ as follows
\be
\mathcal{H}_k \triangleq \{\Delta \qh_k | \Delta \qh_k \in \mathcal{CN}(0,\sigma_h^2\qI)\}.
\ee
Due to the uncertainty introduced by the channel estimation error, the deterministic constraints in \textbf{P1} with perfect CSI become more difficult to satisfy. The key reason for choosing the probabilistic channel estimation error model over the worst-case channel estimation error model is that, worst-case channel estimation models will lead to a worst-case study for the robust beamforming problems, where the worst-case study usually provides a very conservative performance. This is because in practical systems, in order to bound the estimation errors with a threshold, this threshold value might need to be very conservative, so that all possible estimation errors (including those extremely rare cases) are bounded. Therefore in this manuscript, instead of providing a determined bound based on the worst-case scenario, we aim to provide a statistical guarantee to the EH and SINR constraints.

Therefore, we consider the constraints in a statistical manner, where the original deterministic constraints are statistically guaranteed with a probability. By rewriting the corresponding terms of \textbf{P1} in a probabilistic form, the robust formulation of the SWIPT problem under the SINR, total power and SAR constraints with imperfect CSI can be written as
\bea\label{eqn:prob:robust 1}
\textbf{P3:}&& \max_{\{\qw_k, \rho_k,\lambda\}}~ \lambda\\
\mbox{s.t.} && \PP\{\Gamma_k \ge \gamma_k\}\ge \alpha_k,\forall k, \Delta \qh_k \in \mathcal{H}_k, \label{constraint:prob3:SINR} \\
&& \PP\{\Lambda_k \ge \lambda \}\ge \beta_k, \forall k, \Delta \qh_k \in \mathcal{H}_k,\label{constraint:prob3:EH}\\
&& 0\le \rho_k\le 1, \forall k, \notag \\
&& \sum_{k=1}^K \|\qw_k\|^2 \le P_t,\notag\\
&& \sum_{k=1}^K \qw_k^\dag \qA_l \qw_k \leq P_l, \forall l,\notag
\eea
where $\alpha_k$ and $\beta_k$ are the probability guarantees for the SINR and EH constraints, respectively. The robust formulation in \textbf{P3} is non-convex due to the probabilistic constraints, which makes the problem NP-hard and difficult to solve.
 
\subsubsection{The Bernstein-type inequality method}
We first introduce an existing technique in the literature to solve \textbf{P3}, the so-called BTI method. The BTI transforms the probabilistic constraints into a deterministic form based on the large deviation inequality for complex Gaussian quadratic vector functions, which is given in the following lemma \cite{wang2014outage, khandaker2016probabilistically}.

\begin{lemma}\label{Lemma BTI}
If the probabilistic constraint can be represented in the following form
\begin{equation}
\PP\{\qx^{\dagger}\qB\qx + 2 \Re\{\qx^\dag\qz\} + \sigma \ge 0 \} \ge 1-\rho,
\end{equation}
where $\qx$ is a standard complex Gaussian random vector with $\qx \sim \mathcal{CN}(0, \qI)$, $\qB$ is a complex Hermitian matrix, $\qz$ is a complex vector, while the tuple $(\qB, \qz, \sigma)$ forms a set of deterministic optimization variables, and $\rho \in (0,1]$ is fixed, then the following implication holds
\bea
\PP\{\qx^{\dagger}\qB\qx + 2 \Re\{\qx^\dag\qz\} + \sigma \ge 0 \} \ge 1-\rho \\
\Leftarrow \left\lbrace
\begin{matrix}\label{eqn:prob:4}
&\tr(\qB) - \sqrt{-2\ln(\rho) }\psi + \ln(\rho) \psi + \sigma \geq 0, \\
&|| \text{vec}(\qB); \sqrt{2} \qz || \le \psi,\\
&\omega \qI + \qB \succeq \mathbf{0}, \quad \psi, \omega \geq 0,
\end{matrix}
\right.
\eea
where $\psi \in \mathbb{R}$ and $\omega \in \mathbb{R}$ are slack variables, and \eqref{eqn:prob:4} is jointly convex in $\qB$, $\qz$ and $\sigma$.
\end{lemma}

In order to exploit the BTI method, we first apply the SDR with $\qW_k=\qw_k \qw_k^{\dagger}$ as used in the reformulation of \textbf{P2}, and then further rewrite the probabilistic constraints as follows
{\small\bea\label{eqn:prob:4re}
&&\max_{\{\qW_k, \rho_k,\lambda\}}~\lambda\notag\\
\mbox{s.t.}\hspace{-4mm}&&
\PP\bigg\{\Delta\qh_k^\dag \qQ_k \Delta\qh_k + 2\Re\{\Delta\qh_k^\dag\qQ_k\hat\qh_k\} + \hat\qh_k^\dag \qQ_k \hat\qh_k\notag\\
&&\qquad\qquad - \gamma_k\left(N_0 + \frac{N_C}{\rho_k}\right)\ge 0\bigg\}\ge \alpha_k,\forall k, \Delta \qh_k \in \mathcal{H}_k,\notag\\
&& \PP\bigg\{\Delta\qh_k^\dag \qW \Delta\qh_k + 2\Re\{\Delta\qh_k^\dag\qW\hat\qh_k\} + \hat\qh_k^\dag \qW \hat\qh_k\notag\\
&& \qquad\qquad+ N_0 - \frac{\lambda}{1-\rho_k} \ge 0 \bigg\}\ge \beta_k, \forall k, \Delta \qh_k \in \mathcal{H}_k,\notag\\
&& \tr(\qW)\le P_t,\notag \\
&& 0\le \rho_k\le 1, \qW_k\succeq \qzero, \forall k, \notag \\
&& \tr(\qA_l \qW) \leq P_l, \forall l,\notag
\eea}\vspace*{-5mm}

\noindent where we have defined $\qW\triangleq\sum_{k=1}^K \qW_k $ and $\qQ_k \triangleq \qW_k- \gamma_k\sum_{j=1, j\ne k}^K \qW_j$ for notation convenience.
	
Then, by applying the BTI method to transform the probabilistic constraints regarding the SINR and EH, we can get the following robust formulation,
\bea\label{eqn:prob:5}
\textbf{P4:} && \max_{\{\qW_k, \rho_k, x_k, z_k, \mu_k, \nu_k, \lambda\}}~\tilde\lambda \\
\mbox{s.t.} &&
~\sigma_h^2\tr(\qQ_k) - \sqrt{-2\ln(1-\alpha_k)}z_k + \ln(1-\alpha_k)x_k\notag\\
&&\hspace{19mm} + \hat\qh_k^\dag \qQ_k \hat\qh_k - \gamma_k\left(N_0 + \frac{N_C}{\rho_k}\right) \ge 0, \notag\\
&&x_k\qI + \sigma_h^2\qQ_k\succeq \qzero,\notag\\
&& \|\sigma_h^2\text{vec}(\qQ_k); \sqrt{2}\sigma_h\qQ_k\hat\qh_k\| \le z_k, \notag\\
&& \sigma_h^2\tr(\qW) - \sqrt{-2\ln(1-\beta_k)}\nu_k + \ln(1-\beta_k)\mu_k\notag\\
&&\hspace{20mm} + \hat\qh_k^\dag \qW \hat\qh_k + N_0 - \frac{\tilde\lambda^2}{1-\rho_k} \ge 0,\notag\\
&&\mu_k\qI + \sigma_h^2\qW\succeq \qzero,\notag\\
&& \|\sigma_h^2\text{vec}(\qW); \sqrt{2}\sigma_h\qW\hat\qh_k\| \le \nu_k, \notag\\
&& \tr(\qW)\le P_t,\notag\\
&& 0\le \rho_k\le 1, \qW_k\succeq \qzero, \forall k,\notag\\
&& \tr(\qA_l \qW) \leq P_l, \forall l.\notag
\eea
By the definition of $\qW_k = \qw_k\qw_k^\dag$, $\qW_k$ should be of rank one, which has been relaxed to be positive semidefinite in P4. Since P4 is a convex problem, its solutions $\qW_k$ can be efficiently obtained via numerical software packages such as CVX, which are always optimal to the transformed problem \textbf{P4}, but not necessarily optimal to the original robust formulation of the SWIPT problem P3. If the solution of $\qW_k$ is of rank one, then it is also the optimal solution to the original problem \textbf{P1}. In such cases, the solution of $\qw_k$ can be derived based on $\qW_k$, where it is straightforward when the rank of $\qW_k $ is one. For the cases where the rank of $\qW_k $ is higher than one, the near-optimal solution of $\qw_k$ can be derived via the rank-one approximation of the $\qW_k $, e.g. via singular value decomposition methods \cite{ZL}. Since the transformed deterministic constraints are convex in $\qW_k, \rho_k, x_k, z_k, \mu_k$ and $\nu_k$, the original robust beamforming for SWIPT in \textbf{P3} has been transformed into a convex problem as in \textbf{P4}. Note that by using the implication in BTI in Lemma \ref{Lemma BTI}, the transformed constraints in \textbf{P4} characterize the lower bounds on the probability $\alpha_k$ and $\beta_k$. Therefore, the feasible solutions of \textbf{P4} are sub-optimal solutions of the original NP-hard problem \textbf{P3}, and can be efficiently solved using convex optimization solvers such as CVX, but the performance could be conservative.
	
\subsubsection{Deep Learning Based Method}\label{subsection deep learning based method}
	
\begin{figure}[t]\centering
  \includegraphics[width=\linewidth]{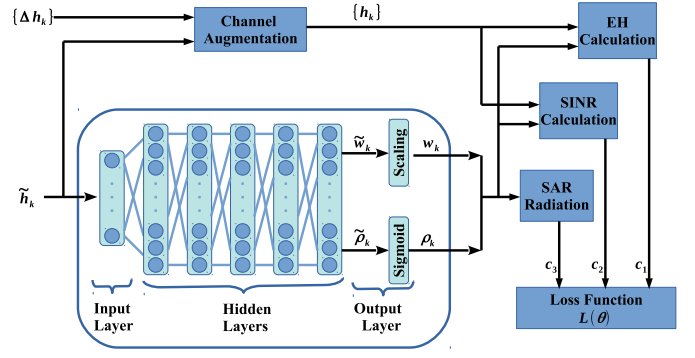}
  \caption{The proposed data-augmentation based training method for the robust beamforming problem \textbf{P1} for SWIPT under SINR, total power and SAR constraints.}\label{fig NN structure}
\end{figure}

Based on the above analysis, conventional solutions such as the BTI method might transform the probabilistic constraints to a more tractable form, but at the cost that the solutions are expected to be conservative comparing to the optimal solutions. Instead of seeking sub-optimal solutions via the conventional techniques like BTI, in this subsection, we will exploit the power of the deep learning neural networks (NNs), which learn from the data to form the robust beamforming strategies for the studied SWIPT under SINR, total power and SAR constraints. A general design of the training method is illustrated in Fig. \ref{fig NN structure}. Specifically, we will present the data-augmentation based technique to transform the probabilistic constraints to an NN training problem, and then reformulate the optimization problem with multiple constraints to a deep learning problem. Note that our method can deal with general channel error distributions and is not limited to Gaussian errors.
	
The general objective of the deep learning method is to train an NN with the estimated channels $\hat \qh_1, \dots, \hat \qh_K$ as inputs, and it will output the beamforming solutions of \textbf{P3} as follows
\begin{equation}
\{\qw_1, \dots, \qw_K, \rho_1, \dots, \rho_K \} = f(\hat \qh_1, \dots, \hat \qh_K;\qtheta),
\end{equation}
where $\qtheta$ denotes the parameter set of the NN. In this way, it transforms the original non-convex optimization problem about the beamforming $\{\qw_k\}$ to finding the optimal parameter set $\qtheta$. However, note that the deep learning method cannot automatically address the constraints, which should be addressed by a deliberate design of the training procedure.

According to the SINR definition in \eqref{eq gamma}, the calculation of $\Gamma_k$ requires the estimated CSI $\qh_k$, the beamforming vector $\qw_j$ for $j=1,\dots, K$, and the PS parameter $\rho_k$. Without loss of generality, we can rewrite $\Gamma_k$ as a function of the NN $f(\hat \qh_1, \dots, \hat \qh_K;\qtheta)$ as follows:
\begin{align}\label{eq gamma nn}
&\Gamma_k(\hat \qh_k, \Delta \qh_k, f(\hat \qh_1, \dots, \hat \qh_K;\qtheta))\nonumber\\
&\quad\qquad\qquad\triangleq \frac{{\rho_k}|\qh_{k}^\dag\qw_k|^2}{ {\rho_k}\left(\sum_{j=1,j \ne	k}^K |\qh_k^\dag\qw_j|^2 + N_0\right) + {N_C}}.
\end{align}
Similarly, the EH power of receiver $k$ in \eqref{eq lambda} can be rewritten as
\begin{align}\label{eq lambda nn}
\Lambda_k(\hat \qh_k, \Delta \qh_k, f(&\hat \qh_1, \dots, \hat \qh_K;\qtheta))\nonumber\\
&\triangleq (1-\rho_k) \left( \sum_{j=1}^K |\mathbf{h}_k^{\dagger} \mathbf{w}_j|^2 + N_0 \right).
\end{align}
The SAR function can be also rewritten as follows
\begin{equation}
S_l(f(\hat \qh_1, \dots, \hat \qh_K;\qtheta)) \triangleq \sum_{k=1}^K \qw_k^\dag \qA_{l} \qw_k.
\end{equation}
	
\paragraph{Problem Reformulation via Quantile Functions}
Similar to the analytical studies in the previous sections, it is also a challenging problem to address the probabilistic constraints in the NN. To facilitate the evaluation of the probabilistic constraints, we first introduce the quantile function $q(x, \sigma )$ as follows,
\begin{equation}
q(x, \sigma) = \inf \{z | \PP\{x \leq z\} \leq \sigma\},
\end{equation}
where the quantile function $q(x, \sigma)$ returns the quantile value (infimum value) $z$, such that for all $x$, the probability $\PP\{x \leq z\}$ is no more than the value of $\sigma$. In this way, given a probabilistic constraint in the following form
\begin{equation}
\PP\{x \leq Z\} \leq \sigma,
\end{equation}
then it can be rewritten by the quantile function in the equivalent form as follows
\begin{equation}
Z - q(x, \sigma ) \leq 0,
\end{equation}
which transforms the comparison from the probability $\PP\{x \leq Z\}$ against $\sigma$, to the quantile value $q(x, \sigma )$ against the threshold $Z$. Therefore, for the SINR constraint \eqref{constraint:prob3:SINR} of receiver $k$, it can be first rewritten into the following form
\begin{equation}
\PP\{\Gamma_k \le \gamma_k\}\le 1-\alpha_k.
\end{equation}
Then, with the quantile function, it can be further transformed as follows
\begin{equation}
\gamma_k - q(\Gamma_k(\hat \qh_k, \Delta \qh_k, f(\hat \qh_1, \dots, \hat \qh_K;\qtheta)), 1-\alpha_k ) \leq 0,
\end{equation}
where the NN representation of $\Gamma_k$ in \eqref{eq gamma nn} has been used.

Similarly, the EH constraint \eqref{constraint:prob3:EH} of receiver $k$ can be transformed as follows
\begin{equation}
\lambda - q(\Lambda_k(\hat \qh_k, \Delta \qh_k, f(\hat \qh_1, \dots, \hat \qh_K;\qtheta)), 1-\beta_k ) \leq 0.
\end{equation}

\setcounter{equation}{37}
\begin{figure*}[t!]{\small\begin{align}\label{eq loss function}
		\mathcal{L}(\qtheta) &= -c_1 \min \E_{\hat \qh_k, \Delta \qh_k, \forall k} \{q(\Lambda_k(\hat \qh_k, \Delta \qh_k, f(\hat \qh_1, \dots, \hat \qh_K;\qtheta)), 1-\beta_k)\}\nonumber\\
		&\quad + c_2 \sum_{k=1}^K \left(\E_{\hat \qh_k, \Delta \qh_k, \forall k} \{\gamma_k - q(\Gamma_k(\hat \qh_k, \Delta \qh_k, f(\hat \qh_1, \dots, \hat \qh_K;\qtheta)), 1-\alpha_k)\} \right)_{+} + c_3 \sum_{l=1}^L \left(S_l(f(\hat \qh_1, \dots, \hat \qh_K;\qtheta)) - P_l\right)_{+}.
		\end{align}}
	\setcounter{equation}{39}
	{\small\begin{align}\label{eq gd update theta}
		\qtheta^{(t)} &= \qtheta^{(t-1)} +\frac{c_1}{|\mathcal{S}_{\tilde \qh_k}|} \sum_{\tilde \qh_k \in \mathcal{S}_{\tilde \qh_k}} \frac{1}{|\mathcal{S}_{\Delta \qh_k}|} \sum_{\Delta \qh_k \in \mathcal{S}_{\Delta \qh_k}} \nabla_{\qtheta} \min \{q(\Lambda_k(\hat \qh_k, \Delta \qh_k, f(\hat \qh_1, \dots, \hat \qh_K;\qtheta)), 1-\beta_k)\}\nonumber\\
		&\quad- \sum_{k=1}^K \frac{c_2}{|\mathcal{S}_{\tilde \qh_k}|} \sum_{\tilde \qh_k \in \mathcal{S}_{\tilde \qh_k}} \frac{1}{|\mathcal{S}_{\Delta \qh_k}|} \sum_{\Delta \qh_k \in \mathcal{S}_{\Delta \qh_k}} \nabla_{\qtheta} \left(\gamma_k - q(\Gamma_k(\hat \qh_k, \Delta \qh_k, f(\hat \qh_1, \dots, \hat \qh_K;\qtheta)), 1-\alpha_k)\right)_{+}\nonumber\\
		&\quad- \sum_{l=1}^L \frac{c_3}{|\mathcal{S}_{\tilde \qh_k}|} \nabla_{\qtheta} \left( S_l(f(\hat \qh_1, \dots, \hat \qh_K;\qtheta)) - P_l\right)_{+}.
		\end{align}}\hrulefill\end{figure*}
\setcounter{equation}{32}
	
By noticing $\lambda$ is also the objective of the original robust formulation in \textbf{P3}, we further rewrite the robust formulation of SWIPT based on NN $f(\hat \qh_1, \dots, \hat \qh_K;\qtheta)$ as follows
{\small
\bea\label{eqn:prob:robust NN or}
\hspace*{-10mm}\textbf{P5:}&&\hspace{-5mm} \max_{\qtheta} \mathbb{E}_{\hat \qh_k, \Delta \qh_k, \forall k} \min\notag\\
&&\hspace{2mm}\{q(\Lambda_k(\hat \qh_k, \Delta \qh_k, f(\hat \qh_1, \dots, \hat \qh_K;\qtheta)), 1-\beta_k)\}\label{obj:prob5}\\
\hspace*{-10mm}\mbox{s.t.}&&\hspace*{-5mm}\mathbb{E}_{\hat \qh_k, \Delta \qh_k, \forall k} ~ \min\notag\\
&&\hspace*{-5mm}\{\gamma_k - q(\Gamma_k(\hat \qh_k, \Delta \qh_k, f(\hat \qh_1, \dots, \hat \qh_K;\qtheta)), 1-\alpha_k)\} \leq 0,\label{constraint:prob5:SINR} \\
&&\hspace*{-5mm}0\le \rho_k\le 1, \forall k, \label{constraint:prob5:rho} \\
&&\hspace*{-5mm}\sum_{k=1}^K \|\qw_k\|^2 \le P_t, \label{constraint:prob5:totalpower} \\
&&\hspace*{-5mm}S_l(f(\hat \qh_1, \dots, \hat \qh_K;\qtheta)) \leq P_l, \forall l, \label{constraint:prob5:SAR}
\eea}\vspace{-5mm}

\noindent where the mathematical expectations in objective \eqref{obj:prob5} and the SINR constraint \eqref{constraint:prob5:SINR} are with respect to ${\hat \qh_k, \Delta \qh_k, \forall k}$. This is to make the trained NN parameter set $\qtheta$ generally applicable to all possible $\Delta \qh_k$ and estimated CSI inputs $\hat \qh_k$. Note that since $\qw_k$ and $\rho_k$ are the outputs of the NN $f(\hat \qh_1, \dots, \hat \qh_K;\qtheta)$, the deterministic constraints \eqref{constraint:prob5:rho}--\eqref{constraint:prob5:SAR} can be regarded as the constraints on the NN outputs.

\paragraph{Addressing the constraints} During the training procedure, the NN cannot automatically satisfy the constraints. Therefore, for the problems with constraints, e.g., the studied robust beamforming for the SWIPT problem, each constraint needs to be addressed deliberately. Similar to the analysis in the convex optimization, there are no universal solutions to address all constraints. Some simple constraints can be addressed via the design of the output layer of the NN architecture. Here we will use two output layer design techniques to address the constraint \eqref{constraint:prob5:rho} and \eqref{constraint:prob5:totalpower} as follows:
\begin{itemize}
  \item For the PS constraint in \eqref{constraint:prob5:rho}, each $\rho_k$ should be within the range of $[0,1]$. This can be achieved by applying the Sigmoid function in the output layer for each raw output $\tilde \rho_k$ as $\rho_k = \text{Sigmoid}(\tilde \rho_k)$, with $\text{Sigmoid}(x) = \frac{1}{1+e^{-x}}$.
  \item For the total power constraint in \eqref{constraint:prob5:totalpower}, the sum power of all $\qw_k$, calculated by $\sum_{k=1}^K \|\qw_k\|^2$, should be bounded by the total power $P_t$. This can be addressed by scaling the raw beamforming vectors $\tilde \qw_k$ as $\qw_k = Y \tilde \qw_k$, with $Y = \min\left\{1, \frac{P_t}{\sum_{k=1}^K \| \tilde \qw_k\|^2}\right\}$.
\end{itemize}

With the above output layer design, the constraints \eqref{constraint:prob5:rho} and \eqref{constraint:prob5:totalpower} are enforced by the NN architecture automatically, i.e., all outputs will satisfy both constraints. Therefore, when possible, it is preferable to exploit the NN architectures, e.g., the NN output layer, which firmly address the constraints. However, the above techniques can only address simple constraints such as \eqref{constraint:prob5:rho} and \eqref{constraint:prob5:totalpower}. For complicated ones, such as the probabilistic constraint \eqref{constraint:prob5:SINR} and the SAR constraint \eqref{constraint:prob5:SAR}, we exploit a general technique by modifying the objective function, so that the trained NN parameter set $\qtheta$ should learn to satisfy both via the training procedure. This is achieved by considering the penalty of violating the constraints in the loss function given by \eqref{eq loss function}, where $c_1$, $c_2$ and $c_3$ are positive weight parameters for each individual learning objective terms\footnote{The values of $c_1$, $c_2$ and $c_3$ are empirically determined by trial and error.}, and $(x)_{+} \triangleq \max\{x, 0\}$ is the clamp operation. The use of the clamp operation is to make the loss function $\mathcal{L}(\qtheta)$ increase only when the corresponding constraints are violated.

In this way, we can rewrite the robust beamforming for SWIPT under SINR, total power and SAR constraints as the following unconstrained deep learning problem:\setcounter{equation}{38}
\begin{equation}\label{eq prob6}
\textbf{P6:\quad} \min_{\qtheta} \mathcal{L} (\qtheta),
\end{equation}
where \textbf{P6} provides a transformed formulation of \textbf{P3} that can be solved by deep learning methods. Specifically, \textbf{P3} is first transformed to \textbf{P5} with the help of the quantile function in (27), which is then transformed to \textbf{P6} by integrating the constraints together with the objective function as the loss function that can be used for deep learning. Note that since the loss function $\mathcal{L}(\qtheta)$ exploits the end performance (the expected harvested power) to evaluate the outputs of the NN, the training process can apply the unsupervised training method, where there is no need to know the optimal beamforming vectors and the PS parameters as supervised labeled data in the supervised training method.
	
\paragraph{Data-augmentation Based Training Method}
In the studied robust beamforming problem of SWIPT with SINR, total power and SAR constraints, the most challenging task is to address the probabilistic constraints with regards to the EH and SINR constraints. In the previous sections, we have exploited the quantile function to transform the constraint with probability to a constraint with quantile values. Further, the loss function has been modified so that it is expected that these constraints can be learned by the NN. It can be observed by the loss function $\mathcal{L} (\qtheta)$ in \eqref{eq loss function}, that the channel estimation errors $\Delta \qh_k$ is only required to evaluate the outputs of the NN with regard to the probabilistic constraints, while the NN only calculates outputs based on the estimated CSI $\tilde \qh_k$.
	
Inspired by this, we can add an auxiliary module during the NN training procedure, where each estimated CSI $\tilde \qh_k$ is augmented by a set of channel estimation errors $\mathcal{S}_{\Delta \qh_k}$ to form a set of potential actual CSI $\mathcal{S}_{\qh_k} = \{ \qh_k | \qh_k = \tilde \qh_k + \Delta \qh_k, \forall \Delta \qh_k \in \mathcal{S}_{\Delta \qh_k}\}$, which is then used to evaluate the probabilistic related terms in the loss function $\mathcal{L} (\qtheta)$. The evaluation over the augmented set $\mathcal{S}_{\qh_k}$, provides an estimated performance of the loss function $\mathcal{L} (\qtheta)$, while it is expected to converge to the actual performance when the size of the set increases according to the law of large numbers. Similarly, the mathematical expectation calculation in \eqref{eq loss function} against the estimated CSI $\tilde \qh_k$ can be achieved by an evaluation of the loss function over a set of estimated CSI $\mathcal{S}_{\tilde \qh_k}$. With the loss function $\mathcal{L} (\qtheta)$ estimated over the estimated CSI set $\mathcal{S}_{\tilde \qh_k}$ and the augmented set $\mathcal{S}_{\qh_k}$ for each element in $\mathcal{S}_{\tilde \qh_k}$, the NN parameter set $\qtheta$ can be updated based on the gradient descent method \cite{sutskever2013importance} as \eqref{eq gd update theta}.

When the offline training phase is completed, the trained NN can be deployed for the application and the beamforming solution can be obtained by using the inference mode of the NN, i.e. the channels are used as inputs to the NN and the NN only performs forward calculations without a back-propagation process. Since matrix multiplication is the most computation-intensive operation, we estimate the computational complexity based on the required multiplications, while other operation time is ignored. If the NN has $L$ fully connected layers and the $l$-th layer has $\theta_l$ neurons, then the total computational complexity of the NN-based method can be estimated as $\mathcal{O}(\sum_{l=2}^L \theta_{l-1}\theta_l)$. Besides the fixed number of neurons for the layers $l-2$ to $L-1$, the first layer input is of dimension $2KN_t$ and the last layer output is of dimension $2K(N_t+1)$. Therefore, the total computational complexity of the NN for the proposed method can be estimated as $\mathcal{O}(KN_t)$. The complexity of solving the general problems \textbf{P2} and \textbf{P4} is dominated by the SDP constraints and according to \cite[6.6.3]{Nemirovski}, the associated complexity of the interior-point algorithm for solving these two problems is $\mathcal{O}\left(\sqrt{KN_t}\left(K^3N_t^2+K^2N_t^3\right)\right)$.

\subsection{Simulation Results}
To evaluate the performance of the proposed algorithm, simulation results are presented and discussed in this section. The considered MISO downlink is composed of a transmitter with $N_t=3$ transmit antennas and $K=2$ receivers. Each transmit and receive antenna has $8$ dBi and $3$ dBi gain, respectively. Receivers can harvest energy at frequency $f=915$ MHz, and are randomly located around the transmitter with a distance $l_k\sim U(1,5)$ m and at a direction $\zeta_k \sim U(-\pi,\pi)$. We adopt Rician fading with a Rician factor of $0.5$ to model the channel, due to the short distance between the transmitter and the receivers and the dominance of the line-of-sight (LOS) signal; the path loss coefficient is $2.5$. We consider one SAR constraint, i.e., $L=1$, $P_t=2$ W, $N_0=-70$ dBm and $N_C=-50$ dBm. The SAR matrix is given below by \cite{Ying-CISS-13}\setcounter{equation}{40}
\be
\qA = \left[
\begin{array}{ccc}
	0.35 & -0.64-0.15j & -0.17+0.32j \\
	-0.64+0.15j & 2.51 & -0.31+0.29j \\
	-0.17-0.32i & -0.31-0.29j & 2.32
\end{array}
\right].
\ee

\begin{figure}[t]\centering
	\includegraphics[width=0.35\textwidth]{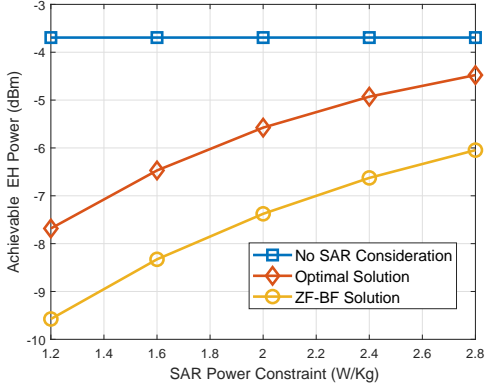}
	\caption{The harvested energy versus the SAR constraints.}\label{fig:EH:vs:SAR}
\end{figure}

\begin{figure}[t]\centering
	\includegraphics[width=0.35\textwidth]{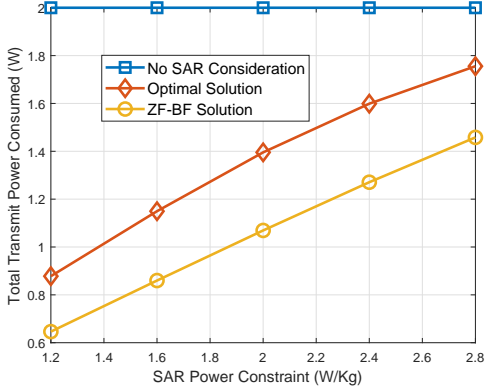}
	\caption{The total transmit power consumption versus the SAR constraints.}\label{fig:pow:vs:SAR}
\end{figure}

\begin{figure}[t]\centering
	\includegraphics[width=0.35\textwidth]{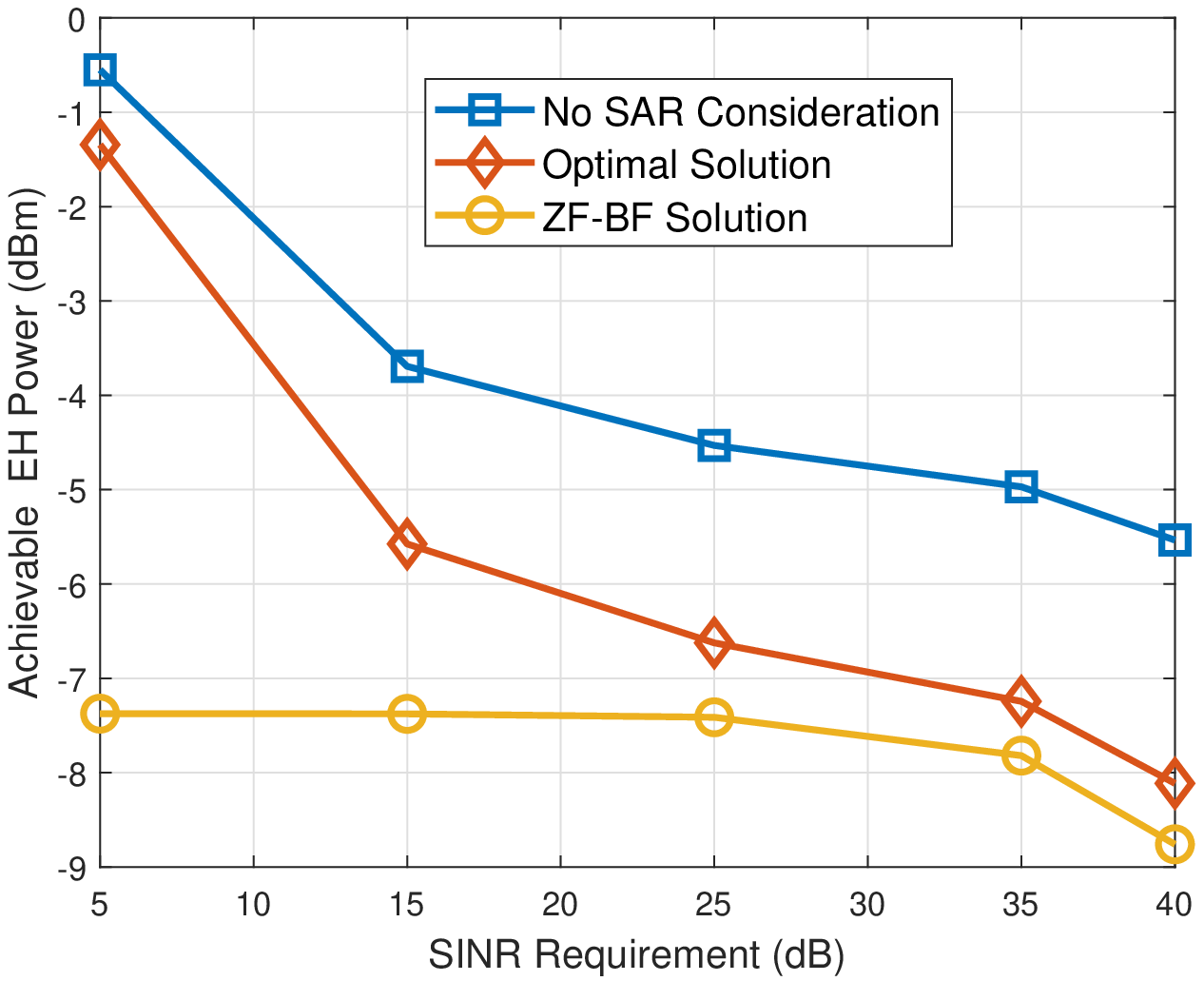}
	\caption{The harvested energy versus the SINR constraints.}\label{fig:EH:vs:SINR}
\end{figure}

\begin{figure}[t]\centering
	\includegraphics[width=0.35\textwidth]{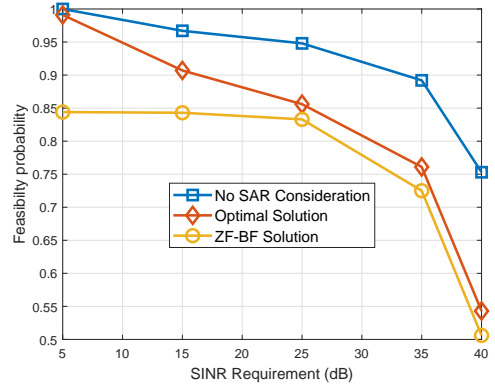}
	\caption{The feasible probability versus the SINR constraints.}\label{fig:fea:vs:SINR}
\end{figure}

\subsubsection{Optimal Beamforming with Perfect CSI}
For performance comparison with the proposed optimal solution with perfect CSI, we consider as benchmarks the zero-forcing beamforming (ZF-BF) and the solution without the SAR consideration (i.e., the solution of \textbf{P1} without considering the SAR constraint). Specifically, the ZF-BF vector is given by $\qw_k=\sqrt{p_k}\tilde\qw_k$, where $\tilde\qw_k=\frac{\left(\mathbf{I}_{N_t} - \mathbf{H}_k^\dag \mathbf{H}_k \right)\mathbf{h}_k}{\left\|\left(\mathbf{I}_{N_t} -\mathbf{H}_k^\dag \mathbf{H}_k \right)\mathbf{h}_k \right\|}$, and $p_k$ is the power for the $k$-th receiver. Let $G_{k,j}\triangleq|\mathbf{h}_{k}^{\dag}\tilde\qw_j|^2$ denote the equivalent link gain between the transmitter and the $k$-th receiver, which satisfies $G_{k,j}=0, \forall k\ne j$. Finally, let $F_{k,l} = {\tilde\qw_k}^\dag \qA_{l} \tilde\qw_k$ denote the $l$-th radiation channel gain due to the transmission intended for the $k$-th receiver.
\bea \label{eq:ZF}
\textbf{P7:} && \displaystyle{\max_{\{p_k\ge 0, \rho_k\}}} \lambda\\
\mbox{s.t.} && \frac{\rho_k G_{k,k} p_k}{\rho_k N_0 + N_C} \geq \gamma_k, \forall k,\notag\\
&& (1-\rho_k)(G_{k,k}p_k+N_0) \geq \lambda , \forall k,\notag\\
&& 0 \le \rho_k \le 1, \forall k,\notag\\
&&\sum_{k=1}^Kp_k \le P_t,\notag\\
&&\sum_{k=1}^K p_k F_{k,l} \leq P_l, \forall l.\notag
\eea
This is a linear programming problem and can be solved using CVX. But as will be illustrated by simulation results, the ZF-BF solutions are very conservative in the EH performance.

Fig. \ref{fig:EH:vs:SAR} shows the harvested power versus the SAR power constraint $P_l$, where the SINR requirement at the receivers is $15$ dB. The proposed optimal solution achieves higher harvested power as $P_l$ increases, and significantly outperforms the ZF-BF scheme. For instance, when $P_l=2$ W, the harvested energy of the proposed solution is about $-5.6$ dBm, which is $1.6$ dB greater than that of the ZF-BF solution. The solution without a SAR constraint remains constant and achieves a higher harvested power than the proposed solution, but the performance gap decreases by increasing $P_l$. Fig. \ref{fig:pow:vs:SAR} depicts the total power consumption at the transmitter for various $P_l$ when $P_t=2$ W. We can observe that the solution without a SAR constraint consumes the full transmission power ($2$ W) to maximize the harvested power. The proposed algorithm makes more effective use of the transmit power than the ZF-BF solution in order to harvest more power while satisfying the SAR constraint.

Fig. \ref{fig:EH:vs:SINR} shows the harvested power versus the SINR requirements, where the SAR power constraint is $2$ W. The harvested power of all solutions decreases as the SINR constraints increase, and this indicates that more power of the received signal is used for information decoding at the high SINR requirement, due to the nature of PS. Fig. \ref{fig:fea:vs:SINR} demonstrates the feasibility probability of the three solutions in terms of the SINR requirements, and the SAR power constraint is $2$ W. When the SINR requirement is low (i.e., $5$ dB), the feasibility probability of the optimal solution (about $99$\%) is close to that of the solution without SAR constraints and much higher than that of the ZF-BF solution (around $84$\%). In the high SINR regime, this gap reduces as the SINR requirement increases and this is because the ZF-BF becomes nearly optimal so both solutions converge.

\subsubsection{Robust Beamforming with Imperfect CSI}
Simulations are conducted to compare the performance of robust beamforming with the three previously discussed solutions when the available CSI is imperfect, which includes a) the non-robust solutions to the SWIPT formulation \textbf{P2}, without considering the channel estimation errors, which is solved via CVX, b) the BTI solutions to the robust formulation \textbf{P4}, and c) the proposed NN-based solution to the robust formulation \textbf{P6} via the proposed method in Fig. \ref{fig NN structure}. The variance of the channel estimation error is set as $\sigma_\qh^2= 10^{-5}$. Specifically, the SINR and SAR thresholds are the same for all receivers, where $\gamma_k = 5$ dB and $P_l = 1.2$ W/kg, while the total transmit power $P_t=2$ W is used. The probability constraints are set as $\alpha_k=95\%, \beta_k=90\%, \forall k$, unless otherwise specified.

\begin{figure}[t]\centering
	\includegraphics[width=0.5\textwidth,clip]{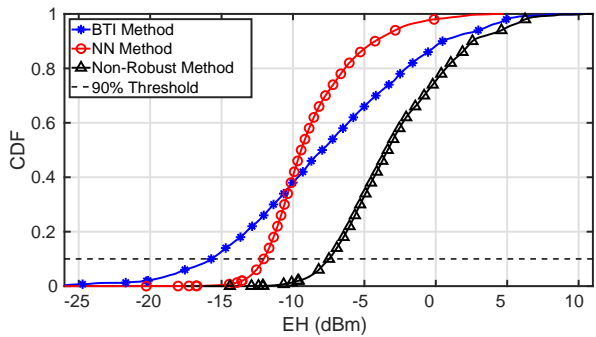}
	\caption{The EH performance comparison between the proposed NN method, the BTI method, and the non-robust method.}\label{fig EH}
\end{figure}

In the proposed NN training method, five fully connected hidden layers are constructed as illustrated in Fig. \ref{fig NN structure}. Each layer has a width of $300$, and the PReLU activation function is used except for the last layer, where $\text{PReLU}(x)$ returns $x$ if $x\ge 0$, or $0.25x$ otherwise. To support the NN training, the Batch normalization technique is used for each layer except for the last layer, and the Adam optimizer is used with a learning rate of $10^{-5}$. Since the proposed NN training method is based on unsupervised training, it only requires to generate the estimated CSI and channel estimation errors for the training purpose. In the simulations, we have randomly generated $4\times10^{5}$ estimated CSIs and exploited the mini-batch training method with a batch size of 4000. During the training procedure, the channel estimation error sets are randomly generated according to its distribution, where each estimated CSI is augmented with 200 channel estimation errors. The weight parameters are empirically selected as $c_1=10000$, $c_2=100$ and $c_3=100$. In addition, a total of 500 testing estimated CSIs are generated, which are used as inputs for the three considered methods, whose beamforming solutions are then evaluated against $10^4$ channel estimation errors for each estimated CSI.
	
We first evaluate the EH performance of the three methods, where the cumulative distribution function (CDF) of the least harvested energy among the receivers is shown in Fig. \ref{fig EH}. It is seen that the NN-based method shows a better EH performance than the BTI method, where the harvested power at the $90\%$ probability threshold is $-12$ dBm and $-15.7$ dBm, respectively. This corresponds to a $3.7$ dB gain for the NN-based method compared to the BTI method. As for the non-robust method, it shows the best EH performance among the three methods, but later it will be shown that the constraints have been violated, which makes the EH performance invalid to use.

\begin{figure}[t]\centering
  \includegraphics[width=0.5\textwidth,clip]{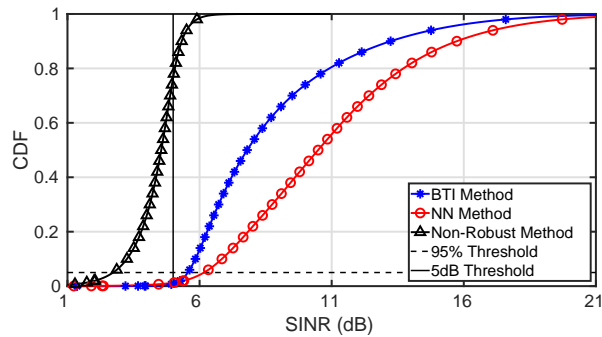}
  \caption{The SINR performance comparison between the proposed DNN method, the BTI method, and the non-robust method.}\label{fig SINR}
\end{figure}
	
The worst receiver's SINR performance is presented in Fig. \ref{fig SINR}, where the probability threshold of 95\% is indicated by the horizontal dashed line, and the 5 dB SINR threshold is indicated by the vertical solid line. Specifically, it is seen that both the NN-based method and the BTI method meet the probabilistic constraints regarding SINR, where their worst receiver's SINR have been maintained above 5 dB bound at the 95\% probability. Meanwhile, the NN-based method shows better performance compared to the BTI method, i.e. it provides a higher worst user's SINR at the probability threshold 95\%. This is due to the fact that the BTI method transforms the original robust formulation to a convex but conservative formulation, where the solutions can be effectively solved but at the cost of the sub-optimal solutions. On the other hand, the NN-based method empirically evaluates the probabilistic bounds during the training, and it learns to target at better EH performance, while satisfying the SINR requirements. As demonstrated in Figs. \ref{fig EH} and \ref{fig SINR}, both the NN-based and BTI methods, satisfy the robust beamforming constraints as in the original problem \textbf{P3}, while the NN-based method provides better solutions with higher EH performance and the BTI solutions are more conservative compared to the NN solutions. It is noticed that the non-robust method fails to satisfy the SINR constraint, and this is because the non-robust method does not consider the channel estimation errors during its formulation. By jointly considering the EH performance in Fig. \ref{fig EH} and the SINR performance in Fig. \ref{fig SINR}, the results of the non-robust method also indicate the importance of the robust formulation, as the uncertainty introduced by the channel estimation errors could compromise the performance, or even make the solutions invalid to use.

\begin{figure}[t]\centering
  \includegraphics[width=0.5\textwidth,clip]{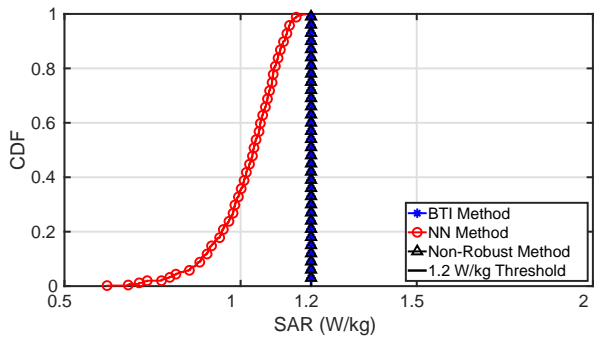}
  \caption{The SAR performance comparison between the proposed NN method, the BTI method, and the non-robust method.}\label{fig SAR}
\end{figure}
	
The SAR performance is shown in Fig. \ref{fig SAR}. It can be seen that the CDF of the BTI method and the non-robust method overlap with the threshold line at $P_l=1.2$ W/kg. This is because the SAR constraint is a linear constraint, while the transformed or relaxed formulations in both the BTI and the non-robust method are in convex forms. Therefore, the overlap with the SAR threshold $P_l=1.2$ W/kg indicates that the solutions via the BTI and non-robust method are based on the equality condition of the SAR constraint. On the contrary, the CDF of the NN method shows that the final SAR performance is within the range from 0.6 to 1.2 W/kg, and the SAR threshold has been met for all tests with regards to both estimated CSI and channel estimation errors. It is also noticed that although the formulation only requires a firm bound on the maximum SAR, the NN-based method can provide solutions with lower SAR radiations.

In summary, the trade-off among the SINR, EH and SAR has been shown to be a non-convex and NP-hard problem in the analytical studies in Section \ref{section robust design}, but the results in Figs. \ref{fig EH}--\ref{fig SAR} demonstrate that the proposed training method has successfully learned better solutions than the BTI method and the non-robust method.

\section{SWIPT with MPE Constraints:\\ Large-Scale Performance Analysis}\label{swipt_mpe}
In this section, we study a SWIPT network from a large-scale point-of-view under the MPE constraint.

\subsection{System model}\label{sec2}

\subsubsection{Network Model}
We consider a large-scale bipolar ad hoc wireless network consisting of a random number of transmitter-receiver pairs. The transmitters form a homogeneous Poisson point process (PPP) $\Phi = \{x_i : i \geq 1 \}$ of density $\la$ in a two dimensional Euclidean space $\R^2$, where $x_i \in \R^2$ denotes the location of the $i$-th transmitter. Each transmitter $x_i$ has a dedicated receiver at a distance $d_0$ in some random direction and transmits with fixed power $P_t$. The time is considered to be slotted and at each time slot all the transmitters are active without any coordination or scheduling. We consider the performance of a receiver located at the origin and its associated transmitter $x_0$. We perform our analysis for this typical receiver but, according to Slivnyak's Theorem \cite{HAE}, our results hold for any receiver in the network.

\subsubsection{Channel Model}
We assume that all wireless links suffer from both small-scale block fading and large-scale path-loss effects. The link between a receiver and an interfering node is in LOS with a probability $p_L$, otherwise it is blocked, e.g. by buildings. On the other hand, the link between a transmitter and its dedicated receiver is always considered to be in LOS, i.e. $p_L = 1$. The interference effect from non-LOS signals is ignored, as we assume the dominant interference is caused by the LOS signals \cite{JEFF}. We consider independent Nakagami fading with parameter $\mu$ for each LOS link and so the power of the channel fading is a normalized gamma random variable with shape parameter $\mu$ and scale parameter $1/\mu$. We denote by $h_i$ the channel gain for the link between the $i$-th transmitter and the typical receiver. Moreover, all wireless links exhibit AWGN with variance $N_0$. The path-loss model assumes that the received power is proportional to $d_i^{-\al}$ where $d_i$ is the Euclidean distance from the origin to the $i$-th transmitter and $\al > 2$ is the path-loss exponent.

\subsubsection{Sectorized Antenna Model}
The transmitters and receivers are equipped with multiple antennas and employ adaptive directional beamforming \cite{THO}. For the sake of analytical tractability, we make use of an approximation of an actual beam pattern using a sectorized model. In particular, the gain of a link between a transmitter and a receiver is a discrete random variable given by \cite{THO}
\begin{equation}
g = \begin{cases}
M^2 &\text{with probability $(\frac{\omega}{\pi})^2$},\\
Mm &\text{with probability $2\frac{\omega}{\pi}(1-\frac{\omega}{\pi})$},\\
m^2 &\text{with probability $(1-\frac{\omega}{\pi})^2$},
\end{cases}
\end{equation}
where $g$ takes into account three parameters: the main lobe beamwidth $\omega \in [0,\pi]$, the main lobe gain $M$, and the side lobe gain $m$. We let $q_i = \{(\frac{\omega}{\pi})^2, 2(\frac{\omega}{\pi})(1-\frac{\omega}{\pi}), (1-\frac{\omega}{\pi})^2\}$ denote the probability of a link with gain $g_i = \{M^2, Mm, m^2\}$. Finally, we assume that the link gain between each transmitter and its dedicated receiver is equal to $M^2$, i.e. they are perfectly aligned. Due to this sectorization, the PPP $\Phi$ is partitioned into three thinned spatial processes \cite{CP}, as follows: $\Phi_1$ represents the set of interferers with link gain $g_1$ with the typical receiver, i.e. $\Phi_1$ is a PPP with density $\la_1 = q_1\la$. Similarly, $\Phi_2$ and $\Phi_3$ are the set of interferers with link gains $g_2$ and $g_3$, respectively, with the typical receiver; as such, $\Phi_2$ and $\Phi_3$ are PPPs with density $\la_2 = q_2\la$ and $\la_3 = q_3\la$, respectively.

\subsection{Performance Analysis with MPE Constraints}
Based on the considered system model, the SINR at the typical receiver can be written as
\begin{align}
\sinr = \frac{\rho P_0 h_0 d_0^{-\al}}{\rho(N_0 + I)+N_C},
\end{align}
where $P_0 \triangleq g_1 P_t$ and
\begin{align}\label{interference}
I \triangleq \sum_{i=1}^3 P_i \sum_{x \in \Phi_i} \frac{h_x}{d_x^\al},
\end{align}
denotes the aggregate interference generated by the transmitters in $\Phi$ at the typical receiver, with $P_i = g_i P_t$. On the other hand, the instantaneous energy harvested at the typical receiver is given by \eqref{eqn:nonlinear} with
\begin{align}
P_r = P_0 h_0 d_0^{-\al} + \sum_{i=1}^3 P_i \sum_{x \in \Phi_i} \frac{h_x}{d_x^\al},
\end{align}
which is the aggregate received signal power at the receiver. Any potential RF EH from the AWGN noise is considered to be negligible.

In this section, we focus on the MPE of the network, evaluated at the origin. In other words, we would like to study the probability of satisfying the MPE constraint $\tau$, expressed as
\begin{align}
&p_s(\tau) = \PP\{\text{MPE} < \tau\},
\end{align}
where
\begin{align}\label{instMPE}
\mpe = \frac{P_0 h_0 d_0^{-\al}}{4\pi d_0^2} + \sum_{i=1}^3 P_i \sum_{x \in \Phi_i} \frac{h_x d_x^{-\al}}{4\pi d_x^2},
\end{align}
which follows from \eqref{mpe} and is referred to as the point source model, where the transmitting antenna is assumed to be represented by a single point source \cite{ITU}. Even though this model does not take into account the antenna size (assumed to be a point), it is accurate in the far-field \cite{ALOUINI,ITU}. We first state the following lemma and then provide the main result.

\begin{lemma}\label{lemma1}
The characteristic function of the interference $I$ is given by
\begin{align}\label{CF}
\phi(t,\la,P,\al) = \exp\left(\frac{2\pi\la}{\al} \left(-\frac{\jmath t P}{\mu}\right)^{\frac{2}{\al}} \mathrm{B}\left(-\frac{2}{\al},\mu+\frac{2}{\al}\right)\right),
\end{align}
where $\lambda$ and $P$ are the density and transmit power of the interfering nodes, respectively.
\end{lemma}

\begin{proof}
The proof is given in Appendix \ref{lemma1_prf}.
\end{proof}

\begin{theorem}\label{MPE_prob}
The probability of satisfying the MPE constraint $\tau$ is given by
\begin{align}\label{thm1}
p_s(\tau) = \frac{1}{2} - \frac{1}{\pi} \int_0^\infty \frac{1}{t}\Im \left\{\frac{\exp(-4\pi\jmath t \tau)\psi(t)}{(1-\jmath t P_0 d_0^{-\al-2}/\mu)^\mu}\right\} dt,
\end{align}
where
\begin{align}
\psi(t) = \prod_{i=1}^3 \phi(t,\la_i,P_i,\al+2),
\end{align}
and $\phi(t,\la_i,P_i,\al+2)$ is given by Lemma \ref{lemma1} with $\la_i = p_L q_i \la$ and $P_i = g_i P_t$.
\end{theorem}

\begin{proof}
The proof is given in Appendix \ref{MPE_prob_prf}.
\end{proof}

Note that as $\mu$ increases, the above probability converges to a constant floor. Specifically, for $\mu \to \infty$, we have
\begin{align}\label{thm11}
\lim_{\mu\to\infty} p_s(\tau) \to \frac{1}{2} - \frac{1}{\pi} \int_0^\infty \frac{1}{t}\Im \left\{\frac{\exp(-4\pi\jmath t \tau)\psi(t)}{1-\jmath t P_0 d_0^{-\al-2}}\right\} dt,
\end{align}
with
\begin{align}\label{thm12}
\lim_{\mu\to\infty} \psi(t) \to \exp\left(\frac{2\pi\Gamma\big(\frac{-2}{\al+2}\big)}{\al+2}(-\jmath t)^{\frac{2}{\al+2}} \sum_{i=1}^3 \la_i P_i^{\frac{2}{\al+2}}\right),
\end{align}
where \eqref{thm11} follows from $(1-x)^a \to 1-ax$ for $x\to0$ and \eqref{thm12} follows from $\Gamma(x+a) \to \Gamma(x)x^a$ for $x \to \infty$.

Next, we are interested in studying both information and energy coverage probabilities when the MPE constraint is satisfied, i.e. the joint probabilities. Due to the correlation between the SINR and the energy harvested with MPE, we assume these events are independent for the sake of analytical tractability. Therefore, we consider the bounds
\begin{align}
\!\!\!\PP\{\mpe < \tau, \sinr > \gamma\} \geq \PP\{\mpe < \tau\}\PP\{\sinr > \gamma\},
\end{align}
and
\begin{align}
\PP\{\mpe < \tau, \e > \epsilon\} \leq \PP\{\text{MPE} < \tau\}\PP\{\e > \epsilon\},
\end{align}
where $\gamma$ and $\eps$ are non-negative thresholds for the SINR and the average harvested energy, respectively, and $\PP\{\sinr > \gamma\}$ and $\PP\{\e > \epsilon\}$ are given in the following corollary.

\begin{corollary}\label{inf_out_prob}
The information coverage probability is given by
\begin{align}
p_o(\gamma) = \frac{1}{2} - \frac{1}{\pi} \int_0^\infty \frac{1}{t}\Im \left\{\frac{\exp(\jmath t (N_C/\rho+N_0))\psi(t)}{(1+\jmath t P_0 d_0^{-\al}/(\gamma\mu))^\mu}\right\} dt,
\end{align}
and the energy coverage probability is given by
\begin{align}
p_e(\epsilon) = \frac{1}{2} + \frac{1}{\pi} \int_0^\infty \frac{1}{t}\Im \left\{\frac{\exp(-\jmath t \delta)\psi(t)}{(1-\jmath t P_0 d_0^{-\al}/\mu)^\mu}\right\} dt,
\end{align}
where $\psi(t) = \prod_{i=1}^3 \phi(t,\la_i,P_i,\al)$ and
\begin{align}\label{delta}
\delta = \frac{c\epsilon}{(1-\rho)(a-\epsilon-b/c)}.
\end{align}
\end{corollary}

The proof for the above expressions follows similar steps as the proof of Theorem \ref{MPE_prob} and thus it is omitted. It is important to point out that in \eqref{delta}, we need $\eps < \bar{a}- \bar{b}/\bar{c}$, that corresponds to the value at which the harvesting circuit saturates. We now derive the joint distribution of the MPE, the SINR and the average harvested energy. In other words, we evaluate
\begin{align}
\PP\{\mpe < \tau,& ~\sinr > \gamma, \e > \eps\}\nonumber\\
&\approx \PP\{\mpe < \tau\}\PP\{\sinr > \gamma, \e > \eps\}.
\end{align}

\begin{theorem}\label{joint_ccdf}
The joint information and energy coverage probability is given by
\begin{align}\label{thm3}
&p_J(\gamma, \eps) = \frac{1}{\pi\Gamma(\mu)}\int_0^\infty \frac{1}{t}\Im \Bigg\{\Bigg(\frac{\Gamma(\mu,\xi(\mu-\jmath t P_0 d_0^{-\al}))}{\exp(\jmath t \delta) (1-\jmath t P_0 d_0^{-\al}/\mu)^\mu}\nonumber\\
&-\frac{\exp(\jmath t (N_C/\rho+N_0)) \Gamma(\mu,\xi(\mu+\jmath t P_0 d_0^{-\al}/\gamma))}{(1+\jmath t P_0 d_0^{-\al}/(\mu\gamma))^\mu}\Bigg)\psi(t)\Bigg\} dt,
\end{align}
where $\psi(t) = \prod_{i=1}^3 \phi(t,\la_i,P_i,\al)$, $\delta$ is given by \eqref{delta} and
\begin{align}\label{low_limit}
\xi = \frac{\gamma d_0^\al}{P_0(1+\gamma)}\left(\delta + N_0 + \frac{N_C}{\rho}\right).
\end{align}
\end{theorem}

\begin{proof}
The proof is given in Appendix \ref{joint_ccdf_prf}.
\end{proof}

\begin{figure}[t]\centering
	\includegraphics[width=0.85\linewidth]{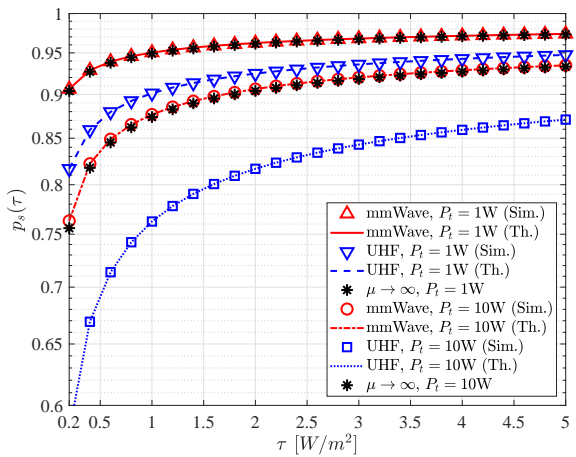}
	\caption{Probability of satisfying the MPE versus the constraint value $\tau$.}\label{fig1}
\end{figure}

Finally, we consider the special case where the interference does not exist. This can be realized in two ways: a) in a blockage dense area, i.e. $p_{\rm L} \to 0$, and b) small beamwidth and small side lobe gain, i.e. $\omega\to 0$ and $m \to 0$. In this case, we have
\begin{align}
&\hat p_J(\gamma, \eps) = \PP\{\mpe < \tau, \sinr > \gamma, \e > \eps\}\nonumber\\ &=\frac{1}{\Gamma(\mu)}\left(\Gamma\left(\mu,\frac{\mu d_0^\al}{P_0}\Xi\right)-\Gamma\left(\mu,\frac{4\pi\tau\mu d_0^{\al+2}}{P_0}\right)\right),
\end{align}
where $\Xi \triangleq \max\left(\gamma N_0+\frac{N_C}{\rho},\delta\right)$. The proof is omitted as it follows directly from the use of the CDF of a gamma random variable. By using the above, we can find the optimal $P_t$ that maximizes the joint distribution. The derivative is expressed as
\begin{align}
\frac{d}{dP_t} \hat p_J(\gamma, \eps) &= \frac{1}{\Gamma(\mu)}\Bigg(\frac{1}{P_t}\left(\frac{\mu d_0^\al\Xi}{M^2P_t}\right)^\mu \exp\left(-\frac{\mu d_0^\al\Xi}{M^2 P_t}\right)\nonumber\\
&- \frac{1}{P_t}\left(\frac{4\pi\tau\mu d_0^{\al+2}}{M^2 P_t}\right)^\mu \exp\left(-\frac{4\pi\tau\mu d_0^{\al+2}}{M^2 P_t}\right)\bigg),
\end{align}
which follows from $d\Gamma(a,x)/dx = -x^{a-1}\exp(-x)$. Then, by setting the derivative equal to zero and solving for $P_t$, we deduce that
\begin{align}
P_t^* = \frac{4\pi\tau d_0^{\al+2} - d_0^\al \Xi}{M^2 \ln \left(\frac{4\pi\tau d_0^2}{\Xi}\right)}.
\end{align}
Observe that the optimal value of $P_t$ is independent of the fading parameter $\mu$. As expected, $P_t^*$ increases with $\tau$ but decreases with $\Xi$ and $M$.

\subsection{Numerical Results}
We now evaluate the performance of SWIPT networks with MPE constraints and validate our analytical framework with Monte Carlo simulations. For the sake of comparison, we consider networks operating in millimeter wave (mmWave) bands (e.g. $30$ GHz) and in ultra high frequencies (UHF) (i.e. $300$ MHz to $3$ GHz). Following a similar approach to \cite{THO}, we use $M = 10$ dB for mmWave and $M=0$ dB for UHF, i.e. the mmWave antenna gain is ten-fold the one of UHF. Moreover, we consider $\mu=5, N_0 = -117$ dB, $p_\text{L} = 0.8$ (mmWave) and $\mu=1, N_0 = -127$ dB, $p_\text{L} = 1$ (UHF). Finally, the remaining parameters are set as: $\omega = \pi/6$, $m = -10$ dB, $d_0 = 5$ m, $\al=3$, $N_C = 0$ dB, $\gamma = -10$ dB, $\eps = -5$ dB and $\rho = 0.5$.

Fig. \ref{fig1} depicts the probability of satisfying the MPE constraint. As expected, the probability decreases with the transmit power and increases with the constraint $\tau$. We can observe that the mmWave band satisfies the constraint with a higher probability, compared to the UHF band with the same transmit power. Indeed, for $P_t = 10$W and $\tau = 0.2$ W/m$^2$, mmWave can satisfy the constraint around $75$\% of the time, whereas the RF exposure with UHF below $\tau$ is less than $60\%$ of the time. This shows the positive impact of directional beamforming, which can be achieved by higher frequencies, on the network's overall RF exposure. The figure also depicts the asymptotic scenario $\mu\to\infty$ (Eq. \eqref{thm11}). It can be seen that this provides a lower bound on $p_s(\tau)$. Finally, the analytical results (lines) and simulation results (markers) are in agreement, which verifies our analysis.

\begin{figure}[t]\centering
  \includegraphics[width=0.85\linewidth]{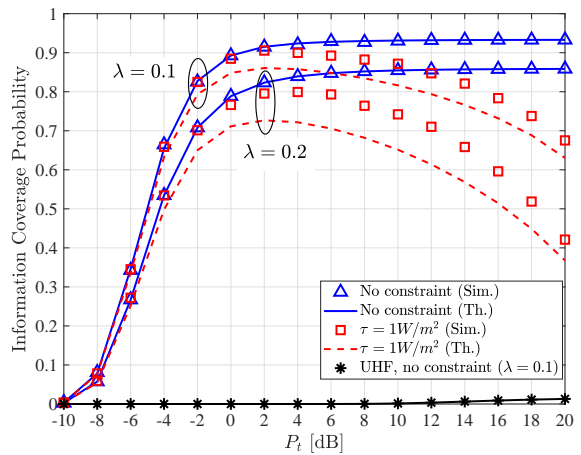}
  \caption{Information coverage probability versus the transmit power $P_t$.}\label{fig2}
\end{figure}

Fig. \ref{fig2} illustrates the information coverage probability in terms of the transmit power. When there is no MPE constraint, the probability converges to a constant ceiling for large values of $P_t$; this corresponds to the interference-limited scenario. Observe that mmWave networks significantly outperform UHF, as also shown in \cite{THO}. On the other hand, when an MPE constraint is imposed, the coverage probability decreases after a certain value of $P_t$ as the MPE level is exceeded more frequently. In other words, there is an optimal value that maximizes the coverage probability, which can be easily deduced by algorithms such as the bisection method. Similar observations can be derived for the energy coverage, shown in Fig. \ref{fig3}. In particular, with no MPE constraints, the probability converges to one for high values of $P_t$, whereas it drops after a certain value of $P_t$ when an MPE level is enforced. However, observe that the optimal $P_t$ here is different compared to the information case. Moreover, the UHF bands perform well in terms of energy coverage but are still outperformed by mmWave bands due to the higher antenna gain at the receivers and transmitters.

\begin{figure}[t]\centering
  \includegraphics[width=0.85\linewidth]{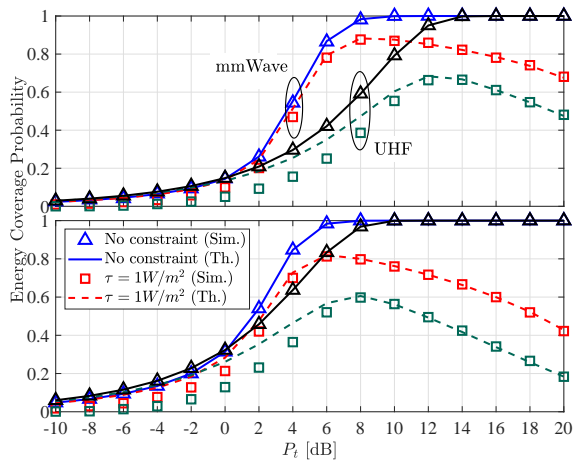}
  \caption{Energy coverage probability versus the transmit power $P_t$; $\la = 0.1$ (upper sub-figure) and $\la = 0.2$ (lower sub-figure).}\label{fig3}\vspace{-2mm}
\end{figure}

Finally, Fig. \ref{fig4} shows the joint coverage probability with respect to the LOS probability. It can be seen that the impact of LOS probability depends on the value of the transmit power. At low $P_t$, the performance improves with $p_\text{L}$ since this facilitates in harvesting more RF energy. Note that for these values, the performance loss from imposing safety constraints is small. On the contrary, for higher values of $P_t$, the coverage probability decreases with $p_\text{L}$ as the effect of interference is detrimental to the SINR. Also, the losses in performance due to MPE constraints are more notable in this case.

\section{Conclusion}\label{conclusion}
In this paper, we provided an framework for the design and analysis of far-field SWIPT under safety constraints. We focused on two RF exposure regulations, the SAR and the MPE, and outlined the state-of-the-art as well as the modeling approach in the context of communication networks. A design for optimal robust beamforming based on deep learning was proposed, subject to specific information, EH and SAR constraints. In addition, a complete theoretical study for the performance of large-scale SWIPT systems under the MPE constraint was derived, with regards to both information and energy coverage. Our results provide insights in terms of the optimal SWIPT design and show the potentials from the proper development of SWIPT systems under health and safety restrictions.\vspace{-2mm}

\begin{figure}[t]\centering
	\includegraphics[width=0.85\linewidth]{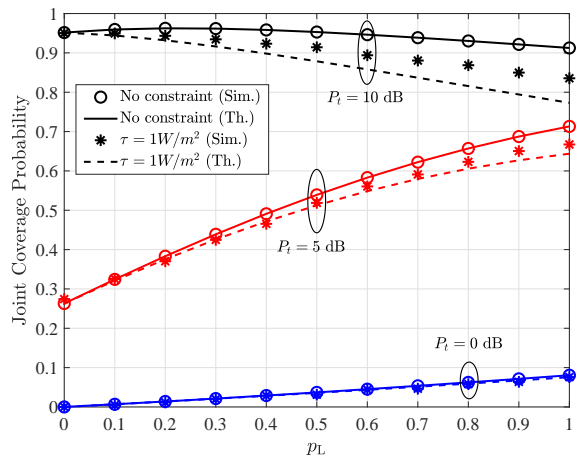}
	\caption{Joint coverage probability versus the LOS probability $p_L$.}\label{fig4}
\end{figure}

\appendix
\subsection{Proof of Theorem \ref{theo2}}\label{prf_thm1}
First the partial Lagrangian of the problem \textbf{P2} is written as
\bea
&& L(\{\qW_k, \rho_k,\alpha_k, \beta_k, \nu_l,\mu\}) = \mu\left(\sum_{k=1}^K \tr(\qW_k) - P_t\right) \notag\\
&&+ \sum_{l=1}^L \nu_l \left(\sum_{k=1}^K \tr(\qA_l \qW_k) - P_l \right) \notag\\
&&+ \sum_{k=1}^K\alpha_k \Bigg(\sum\limits_{j=1}^K \tr(\qh_{k}\qh_{k}^\dag\qW_j) + N_0\notag\\
&&\qquad\qquad\qquad+ \frac{N_k}{\rho_k} - \bigg(1+\frac{1}{\gamma_k}\bigg) \tr(\qh_{k}\qh_{k}^\dag\qW_k)\Bigg) \notag \\
&& + \sum_{k=1}^K\beta_k \left(\tilde\lambda^2 - (1-\rho_k) \left(\sum\limits_{j=1}^K \tr(\qh_{k}\qh_{k}^\dag\qW_j) + N_0 \right)\right) \notag\\
&& = \sum_{k=1}^K \tr(\qW_k \qX_k) - \sum_{l=1}^L \nu_l P_l + \sum_{k=1}^K\alpha_k \bigg(N_0 + \frac{N_i}{\rho_i}\bigg)\notag\\
&&\quad+ \sum_{k=1}^K\beta_k ( \tilde\lambda^2- (1-\rho_k)N_0),\notag
\eea
where $\{\alpha_k, \beta_k, \nu_l,\mu\}$ are dual variables, and we have defined
\begin{align}
\qX_k &\triangleq \mu\qI - \sum_{j=1}^K\beta_j(1-\rho_j) \qh_j\qh_j^\dag + \sum_{j=1}^K\alpha_j \qh_j\qh_j^\dag\nonumber\\
&\quad- \alpha_k \bigg(1+\frac{1}{\gamma_k}\bigg) \qh_{k}\qh_{k}^\dag+ \sum_{l=1}^L \nu_l \qA_l.
\end{align}
So the dual problem is
\bea
&&\max_{\bm\nu, \bm\alpha, \bm\beta, \mu\ge 0} ~ - \sum_{l=1}^L \nu_l P_l + \sum_{k=1}^K\alpha_k \left(N_0 + \frac{N_i}{\rho_i}\right)\notag\\
&&\qquad\qquad\quad+ \sum_{k=1}^K\beta_k (\tilde\lambda^2- (1-\rho_k)N_0)\notag\\
&& \mbox{s.t.} ~ \qX_k = \mu\qI + \sum_{j=1}^K (\alpha_j-\beta_j(1-\rho_j)) \qh_j\qh_j^\dag \notag\\
&&- \alpha_k \left(1+\frac{1}{\gamma_k}\right) \qh_k\qh_k^\dag + \sum_{l=1}^L \nu_l \qA_l\succeq \qzero, \forall k.\notag
\eea
Next, we prove that the matrix $\mu\qI + \sum_{j=1}^K (\alpha_j-\beta_j(1-\rho_j)) \qh_j\qh_j^\dag + \sum_{l=1}^L \nu_l \qA_l$ is full rank by contradiction. If it is not full-rank, suppose there exits a non-zero vector $\qx$ that satisfies $\qx^\dag(\mu\qI + \sum_{j=1}^K (\alpha_j-\beta_j(1-\rho_j)) \qh_j\qh_j^\dag + \sum_{l=1}^L \nu_l \qA_l)\qx=0$. Because $\qX_k\succeq \qzero$, we have
\begin{align}
\qx^\dag\qX_k\qx &= \qx^\dag \Bigg(\!\mu\qI + \sum_{j=1}^K (\alpha_j-\beta_j(1-\rho_j)) \qh_j\qh_j^\dag +\! \sum_{l=1}^L \nu_l \qA_l \!\Bigg)\qx \nonumber\\
&\quad-\qx^\dag\bigg(\alpha_k \bigg(1+\frac{1}{\gamma_k}\bigg) \qh_k\qh_k^\dag\bigg) \qx\notag\\
&= - \alpha_k \bigg(1+\frac{1}{\gamma_k}\bigg) |\qh_{k}^\dag\qx|^2 \ge 0.
\end{align}
Therefore, it holds true that $\qh_k^\dag \qx=0, \forall k$. It follows that
\bea
\qx^\dag\left(\mu\qI + \sum_{j=1}^K (\alpha_j-\beta_j(1-\rho_j)) \qh_j\qh_j^\dag + \sum_{l=1}^L \nu_l \qA_l\right)\qx\notag\\
= \qx^\dag\left(\qI + \sum_{l=1}^L \nu_l \qA_l\right)\qx > 0,
\eea
which contradicts the assumption that $\qx^\dag(\mu\qI + \sum_{j=1}^K (\alpha_j-\beta_j(1-\rho_j)) \qh_j\qh_j^\dag + \sum_{l=1}^L \nu_l \qA_l)\qx=0$. Therefore, the matrix $\mu\qI + \sum_{j=1}^K (\alpha_j-\beta_j(1-\rho_j)) \qh_j\qh_j^\dag + \sum_{l=1}^L \nu_l \qA_l$ must be full rank, and the rank of $\qX_k$ is at least $N_t-1$. One Karush-Kuhn-Tucker condition of the problem \textbf{P2} is that $\tr(\qW_k \qX_k)=0$, so the rank of $\qW_k$ is at most 1. This completes the proof.

\subsection{Proof of Lemma \ref{lemma1}}\label{lemma1_prf}
The characteristic function $\phi(t,\la,P,\al)$ of the interference $I = P\sum_{x \in \Phi} h_x d_x^{-\al}$ is given by
\begin{subequations}\begin{align}
&\phi(t,\la,P,\al) = \E\{\exp\left(\jmath t I\right)\} = \E_{\Phi, h_x}\!\left\{\exp\left(\jmath t P\sum\limits_{x\in \Phi} \frac{h_x}{r_x^\al}\right)\right\}\nonumber\\
&= \E_{\Phi} \prod\limits_{x\in \Phi} \E_{h_x}\left\{\exp\left(\jmath t P \frac{h_x}{r_x^\al}\right)\right\}\nonumber\\
&= \exp\left(2\pi\la \int_0^\infty \left(\E_{h_x}\left\{\exp\left(\jmath t P \frac{h_x}{u^\al}\right)\right\} - 1\right)u du \right)\label{cf1}\\
&=\exp\left(2\pi\la \int_0^\infty\left( \left(1-\frac{\jmath t P}{\mu u^\al}\right)^{-\mu} - 1\right)udu\right),\label{cf2}
\end{align}\end{subequations}
where \eqref{cf1} follows from the probability generating functional of a PPP \cite{HAE}; \eqref{cf2} from the moment generating function of a gamma random variable since $h_x$ are independent and identically distributed. By using the transformation $x = -\frac{\jmath t P}{\mu u^\al}$ and the binomial theorem, the integral can be written as
\begin{align}
&\int_0^\infty\left( \left(1-\frac{\jmath t P}{\mu u^\al}\right)^{-\mu} - 1\right)udu\nonumber\\
&=-\frac{1}{\al}\left(-\frac{\jmath t P}{\mu}\right)^\frac{2}{\al} \sum_{k=1}^\mu \binom{\mu}{k} \int_0^\infty\frac{x^{k-1-2/\al}}{(1+x)^\mu} dx\nonumber\\
&=-\frac{1}{\al}\left(-\frac{\jmath t P}{\mu}\right)^\frac{2}{\al} \sum_{k=1}^\mu \binom{\mu}{k} \mathrm{B}\left(k-\frac{2}{\al},\mu-k+\frac{2}{\al}\right),
\end{align}
which follows from \cite[3.194-3]{GRAD}. With the use of the integral representation of the beta function \cite[8.380-1]{GRAD}, we can write $\mathrm{B}(k-\frac{2}{\al},\mu-k+\frac{2}{\al}) = \int_0^1 x^{-\frac{2}{\al}-1}(1-x)^{\mu+\frac{2}{\al}-1} (\frac{x}{1-x})^kdx$. Then, it is easy to show that the above finite sum is equal to $-\mathrm{B}(-\frac{2}{\al},\mu+\frac{2}{\al})$, which completes the proof.

\subsection{Proof of Theorem \ref{MPE_prob}}\label{MPE_prob_prf}
We will derive the probability of satisfying the MPE constraint $\tau$ by applying the Gil-Pelaez inversion theorem \cite{GP}, that is,
\begin{align}\label{gil}
p_s(\tau) &= \PP\{\text{MPE} < \tau\}\nonumber\\
&= \frac{1}{2}-\frac{1}{\pi} \int_0^\infty \frac{1}{t} \Im\left\{\exp(-\jmath t x) \phi(t)\right\} dt,
\end{align}
where $\phi(t)$ is the characteristic function of the expression in \eqref{instMPE} evaluated at $t$. Therefore, we have
\begin{align}
p_s(\tau) &= \PP\left\{\frac{P_0 h_0}{d_0^{\al+2}} + \sum_{i=1}^3 P_i \sum_{x \in \Phi_i} \frac{h_x}{d_x^{\al+2}} < 4\pi\tau\right\}\\
&=\frac{1}{2}-\frac{1}{\pi} \int_0^\infty \frac{1}{t} \Im\left\{\exp(-\jmath t 4\pi\tau) \phi(t)\right\} dt.
\end{align}
Since \eqref{instMPE} is the sum of four independent terms, its characteristic function is given by the product of the characteristic function of each term. For the first term, we have
\begin{align}
\E\left\{\exp\left(\frac{\jmath t P_0 h_0}{d_0^{\al+2}}\right)\right\} = \left(1-\frac{\jmath t P_0}{\mu d_0^{\al+2}}\right)^{-\mu},
\end{align}
which follows from the fact that $h_0$ is a gamma random variable with shape and scale parameters $\mu$ and $1/\mu$, respectively. Finally, the characteristic function of the $i$-th term in the above sum is $\phi(t,\la_i,P_i,\al+2)$, given by Lemma \ref{lemma1}.

\subsection{Proof of Theorem \ref{joint_ccdf}}\label{joint_ccdf_prf}
The joint distribution can be evaluated as
\begin{align}
&p_J(\gamma, \eps) = \PP\{\sinr > \gamma, \e > \eps\}\nonumber\\
&= \PP\left\{\delta- P_0 h_0 d_0^{-\al} < I < \frac{P_0 h_0 d_0^{-\al}}{\gamma}-\frac{N_C}{\rho}-N_0\right\}\nonumber\\
&= F_I\bigg(\frac{P_0 h_0 d_0^{-\al}}{\gamma}-\frac{N_C}{\rho}-N_0\bigg)-F_I\bigg(\delta- P_0 h_0 d_0^{-\al}\bigg).\label{proof1}
\end{align}
By letting $A \triangleq \frac{P_0 d_0^{-\al}}{\gamma}$ and $B \triangleq \frac{N_C}{\rho}+N_0$, $F_I\left(h_0 A - B\right)$ can be expressed using the Gil-Pelaez inversion theorem as
\begin{subequations}\begin{align}
&F_I(h_0 A - B) = \frac{1}{2}- \frac{1}{\pi}\int_0^\infty \frac{1}{t}\nonumber\\
&\times \Im\left\{\exp(\jmath t B)\int_\xi^\infty\exp(-\jmath t A h) f_h(h) dh \phi(t)\right\} dt\label{cdf1}\\
&= \frac{1}{2}-\frac{1}{\pi} \int_0^\infty \frac{1}{t} \Im\left\{\exp(\jmath t B)\chi(A,\xi) \prod_{i=1}^3 \phi(t,\la_i,P_i,\al)\right\} dt\label{cdf2},
\end{align}\end{subequations}
where \eqref{cdf1} uses the probability density function $f_h(h) = \mu^\mu h^{\mu-1}\exp(-\mu h)/\Gamma(\mu)$ of a gamma random variable with parameters $\mu$ and $1/\mu$. The lower limit $\xi$ is derived by considering
\[\delta- P_0 h_0 d_0^{-\al} < \frac{P_0 h_0 d_0^{-\al}}{\gamma}-\frac{N_C}{\rho}-N_0,\]
and solving for $h_0$ gives \eqref{low_limit}. Then, in \eqref{cdf2}, $\phi(t,\la_i,P_i,\al)$ is given by Lemma \ref{lemma1} and
\begin{align}
\chi(A,\xi) &= \int_\xi^\infty\exp\left(-\jmath t A h\right) f_h(h) dh\nonumber\\
&=\left(1-\frac{\jmath P_0 d_0^{-\al}}{\mu}\right)^{-\mu} \frac{\Gamma\left(\mu, \xi (\mu - \jmath t P_0 d_0^{-\al})\right)}{\Gamma(\mu)},
\end{align}
which follows from \cite[3.381-3]{GRAD}. The second term in \eqref{proof1} can be evaluated with a similar way and the result follows after some algebraic manipulations.

\end{document}